\documentclass[11pt]{article} % set 11 font
\usepackage{subcaption}
\usepackage{definitions}
\usepackage{fullpage}
\usepackage{xr-hyper}
\usepackage{amsmath}
\usepackage{diffcoeff,amssymb}
\usepackage{comment}
\usepackage{amssymb}
\usepackage{diagbox}
\usepackage[]{graphicx}
\usepackage{subcaption}
\usepackage{threeparttable}
\usepackage{pdflscape}
\usepackage{geometry}
\usepackage{booktabs}
\usepackage{dsfont}
\usepackage{url}
\usepackage{algpseudocode}
\usepackage[shortlabels]{enumitem}
\usepackage{setspace}
\usepackage{pgfplots}
\usepackage{float}
\usepackage{xurl}
\usepackage[hidelinks]{hyperref}
\usepackage{cleveref}

\usetikzlibrary{arrows.meta, bending}

\DeclareMathOperator*{\argmax}{arg\,max}

\algnewcommand{\algorithmicand}{\textbf{and }}
\algnewcommand{\algorithmicor}{\textbf{or}}
\algnewcommand{\OR}{\algorithmicor}
\algnewcommand{\AND}{\algorithmicand}

\theoremstyle{remark}
\newtheorem{remark}{Remark}[section]

\newcommand{\bfdelta}{\boldsymbol{\delta}}
\newcommand{\bfv}{\boldsymbol{v}}

\newcommand{\bfr}{\boldsymbol{r}}
\newcommand{\bfX}{\boldsymbol{X}}

\crefname{ass}{assumption}{assumptions}
\crefname{prop}{proposition}{Proposition}
\crefname{lem}{lemma}{lemmas}
\crefname{thm}{theorem}{theorems}

\usepackage[margin=1in]{geometry} % 1-inch margin

\begin{document}

\title{The Power of Linear Programming in Sponsored Listings Ranking: Evidence from a Large-Scale Field Experiment
}

\author{
%anonymous for review
{Haihao Lu \thanks{MIT, Sloan School of Management (\texttt{haihao@mit.edu})}}
\and
{Luyang Zhang \thanks{Carnegie Mellon University, Heinz College (\texttt{luyangz@andrew.cmu.edu})}}
\and
{Yuting Zhu \thanks{National University of Singapore, NUS Business School (\texttt{y.zhu@nus.edu.sg})}}
}

\footnotetext[1]{The authors are listed alphabetically and contributed equally to this work. For confidentiality reasons, the company’s identity remains anonymous throughout the paper. The authors thank the company for providing the empirical setting and data. The views expressed in this paper are solely those of the authors and do not necessarily reflect the views of any affiliated organization. All remaining errors are the authors’ own.}
\date{\today}
% \date{This version: December 10, 2021 \\ First version: December 10, 2021}
\maketitle

\begin{abstract}
Sponsored product advertisements constitute a major revenue source for online marketplaces such as Amazon, Walmart, and Alibaba. A key operational challenge in these systems lies in the Sponsored Listings Ranking (SLR) problem, that is, determining which items to include and how to rank them to balance short-term revenue with long-term relevance and user experience. Industry practice predominantly relies on score-based algorithms, which construct heuristic composite scores to rank items efficiently within strict real-time latency constraints. However, such methods offer limited control over objective trade-offs and cannot readily accommodate additional operational constraints. We propose and evaluate a Linear Programming (LP)-based algorithm as a principled alternative to score-based approaches. We first formulate the SLR problem as a constrained mixed integer programming (MIP) model and develop a dual-based algorithm that approximately solves its LP relaxation within 0.1 second, satisfying production-level latency requirements. In collaboration with a leading online marketplace, we conduct a 19-day field experiment encompassing approximately 329 million impressions. The LP-based algorithm significantly outperforms the industry-standard benchmark in key marketplace metrics, demonstrating both higher revenue and maintained relevance. Mechanism analyses reveal that the performance gains are most pronounced when the revenue–relevance tradeoff is stronger. Our framework also generalizes to settings with inventory, sales, or fairness constraints, offering a flexible and deployable optimization paradigm. The LP-based algorithm was deployed in production at our partner marketplace in January 2023, marking a rare large-scale implementation of a mathematically grounded ranking algorithm in real-world online advertising.

\singlespacing
\noindent\emph{\textbf{Keywords}}: online advertising, sponsored listings rankings, linear programming, latency, revenue-relevance tradeoff, field experiment

\end{abstract}

\newpage
% \vspace{0.2cm}
% \setstretch{1.5}
\section{Introduction}\label{sec:intro}
Sponsored product advertisements constitute a primary revenue stream for most major online marketplaces, including Amazon, Walmart, and Alibaba. For instance, Amazon and Google reported \$37.7 billion and \$224.47 billion in sponsored advertising revenues in 2022, respectively.\footnote{\url{https://ir.aboutamazon.com/annual-reports-proxies-and-shareholder-letters/default.aspx} ; \url{https://abc.xyz/investor/ue}}
 In such advertising systems, third-party sellers pay additional commission fees to online marketplaces to increase the visibility of their products, thereby enhancing the likelihood that their items are prominently displayed to consumers.

A particularly prevalent format of sponsored product advertising is the sponsored product listing, which is typically displayed under a \textit{seed product}. In this format, promoted items are interleaved with organically ranked products and may appear in sections such as “Sponsored Related Items,” “Products Related to This Item,” or “You May Also Like,” with platform-specific variations in terminology. Figure~\ref{fig:slrexample} provides an illustrative example from Amazon. Within these listings, both sponsored and non-sponsored products coexist without predetermined slots reserved for sponsorships. This intermixing introduces substantial complexity to determining which items to include in the listings and how to optimally rank them, a problem we refer to as the \textit{Sponsored Listings Ranking} (SLR) problem, which is the primary focus of this study.\label{s3:SLRintro}

\begin{figure}[ht]
\begin{center}
\caption{\centering Sponsored Product Listings Example}
\includegraphics[width=0.6\textwidth]{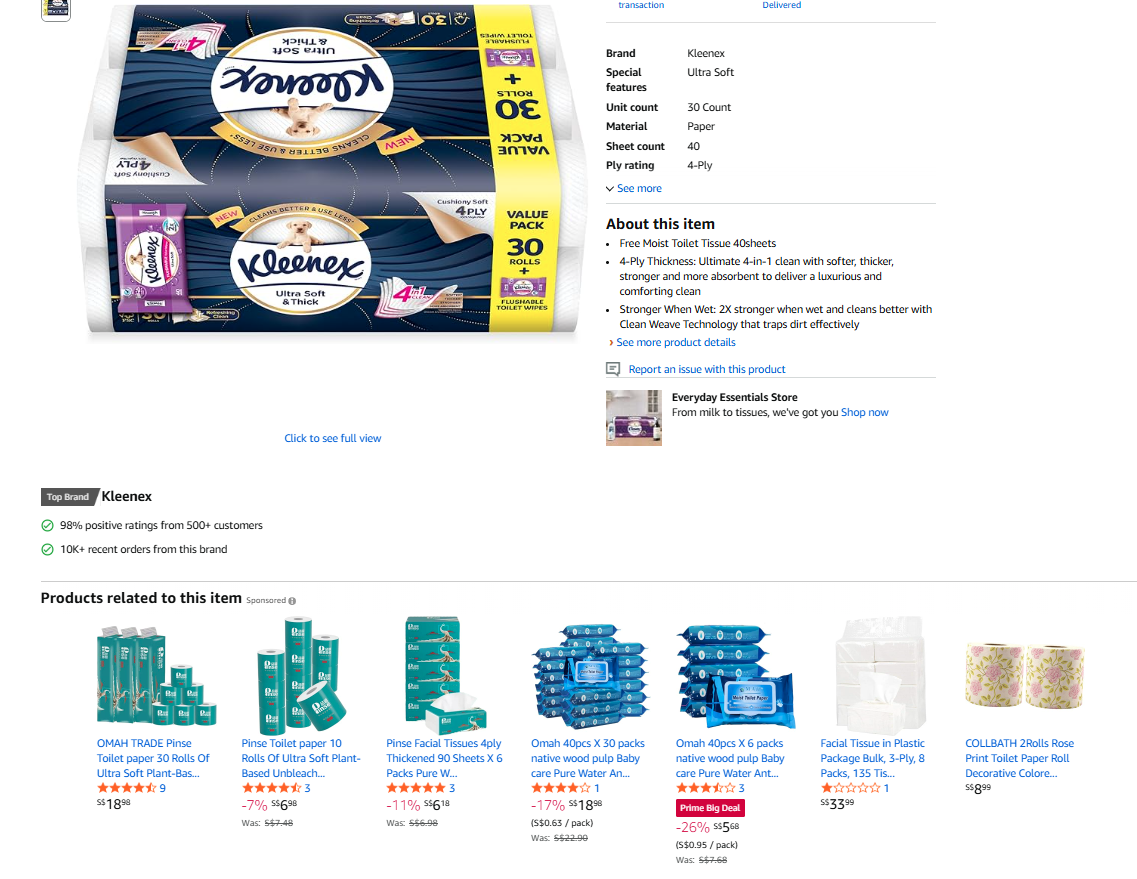}
\label{fig:slrexample}
\end{center}
\end{figure}

The primary challenge in SLR design lies in managing the trade-off between revenue generation and relevance. While displaying more sponsored products can increase short-term platform revenue, such placements may not align with consumer preferences and could degrade user experience. Hence, the platform must balance two competing objectives: maximizing revenue while maintaining the relevance and quality of displayed items to sustain long-term consumer engagement. Achieving this balance requires a principled approach to ranking\label{s3:revenuerelevancetradeoff}.

Beyond this fundamental trade-off, the practical design of SLR algorithms involves several additional considerations. First, in terms of computational efficiency, the algorithm must satisfy stringent latency requirements, typically producing a ranked list within 0.1 seconds to meet real-time display constraints. Second, with respect to interpretability and tractability, the ranking framework should embody clear objectives and explicit constraints, allowing the decision-making process to remain transparent and auditable. Third, the algorithm must be flexible and adaptable, enabling adjustments in objectives, constraints, or weighting schemes to accommodate platform-specific business goals and varying marketplace contexts.

To address practical considerations, particularly the stringent latency requirement, most industry solutions to the SLR problem adopt a “predict-then-optimize” paradigm comprising two stages. In the prediction stage, platforms estimate user action rates, such as click-through rate (CTR) or purchase-through rate (PTR), which correspond to the prevailing Cost-Per-Click (CPC) and Cost-Per-Acquisition (CPA) business models. In the optimization stage, the system determines which sponsored items to display and how to rank them. Our proposed algorithmic solution intervenes at this optimization stage while keeping the prediction stage identical to that used in industry\label{s3:predictthenoptimize}.

A widely adopted approach for the optimization stage in practice is the score-based ranking algorithm, which assigns each item a composite score based on predefined criteria and ranks items according to these scores. These algorithms are favored for their computational efficiency, typically producing rankings within 0.1 seconds to satisfy real-time latency requirements. The scoring criteria generally integrate multiple objectives, such as predicted consumer action rate, item relevance, and potential revenue generation. In most implementations, platforms employ a linear combination of these objectives to construct the score function that best aligns with their business priorities \citep{liu2015convolutional, cheng2016wide, covington2016deep, okura2017embedding, zhou2018deep}.

However, score-based algorithms face several inherent limitations. First, the final ranking quality depends critically on how the score function is specified, making it difficult for marketplaces to select an appropriate formulation. Rankings may vary substantially across different scoring functions, and when the chosen function does not directly represent the underlying economic objectives, the results can be suboptimal and fail to achieve Pareto efficiency. Moreover, aggregating multiple objectives into a single composite score provides no explicit control over the performance of individual objectives. In particular, this approach cannot readily accommodate operational constraints that marketplaces may wish to impose, such as inventory limits, sales targets, or fairness considerations.

To address these challenges, we propose and examine an alternative approach, referred to as the Linear Programming (LP)-based algorithm, in this paper. Specifically, we formulate the ranking problem as a constrained mixed integer programming (MIP) model. Because solving the MIP exactly is NP-hard, we relax it to a tractable LP formulation. However, even with state-of-the-art solvers, directly solving the LP remains too computationally intensive to satisfy the stringent latency requirements of online marketplaces. Motivated by these constraints, we develop a dual-based algorithm that approximately solves the original MIP formulation of the SLR problem within 0.1 second, thereby meeting the platform’s real-time demands. We then validate the proposed approach through a large-scale field experiment conducted in collaboration with a leading online marketplace (referred to as Marketplace A for confidentiality). The 19-day experiment, run in April 2022 and encompassing approximately 329 million impressions, demonstrates that the LP-based algorithm substantially outperforms the highly tuned industry-standard score-based algorithm, leading to significant improvements in key marketplace performance metrics.

We further examine the underlying mechanism through which our LP-based algorithm outperforms the score-based benchmark. The key difference lies in how the revenue–relevance tradeoff is modeled. In the score-based algorithm, the industry convention is to frame this tradeoff as one between \textit{organic revenue} and \textit{ad revenue} (see Section \ref{sec:eBay display} for details on the revenue structure). In contrast, our LP-based algorithm adopts a more disciplined formulation of the revenue–relevance tradeoff: rather than separating revenue by source, it maximizes total marketplace revenue subject to a minimum relevance constraint. Empirically, we provide evidence supporting this mechanism. The performance gains of the LP-based algorithm are more pronounced precisely when the revenue–relevance tradeoff is stronger. Furthermore, the algorithm demonstrates advantages in both the \textit{item selection} and \textit{conditional ranking} stages, consistently identifying items with higher expected revenue.

Our key contribution lies in developing a new ranking algorithm, the LP-based algorithm, for the SLR problem, which can simultaneously achieve higher revenue and satisfy the stringent latency requirements of production environments. We further provide empirical evidence from a large-scale field experiment demonstrating the superior performance of the proposed approach. To the best of our knowledge, this paper represents the first empirical study to both document large-scale experiment evidence and uncover the mechanisms underlying the LP-based algorithm’s advantage over score-based algorithms. In addition, our LP-based framework has strong modeling flexibility, enabling it to incorporate additional operational constraints, such as inventory limits, sales targets, or fairness considerations, within an online optimization setting.\label{s3:contribution}

Our proposed algorithm was deployed in production at Marketplace A in January 2023, where it helps the platform deliver high-value sponsored items to consumers in real time. Notably, score-based algorithms have long served as the industry standard for sponsored listing and other ranking problems. Our work, by contrast, introduces a principled and practically viable alternative that has been successfully implemented in production, demonstrating both theoretical soundness and real-world effectiveness.

%--------------------------------Lit Review----------------------------------%

\subsection{Literature reivew}\label{sec:lit review}

% The majority of previous work regarding sponsored listings ranking problems involves two parts: estimation and ranking. Given that this paper primarily concentrates on the ranking problem, our review will mostly focus on prior studies in this area. Furthermore, while the ranking problem has been extensively investigated within the computer science domain, the closely associated assortment optimization problem has received significant attention in the operations management field, whose related literature will also be reviewed in this section. Last, as online marketplaces might incorporate multi-objectives in the ranking phrase, we also include a summary of past work on the selection of objectives and trade-offs among multi-objectives. 

This section presents related literature on sponsored listing ranking problems.

\paragraph{Ranking in recommendation system}

Our work provides an approach for the ranking step in the recommendation system. In a recommendation system, the first step is to estimate the conversion rate of displaying the impression by machine learning algorithms, which has been quite mature and provide reasonable estimations in the online advertising industry \citep{liu2015convolutional,cheng2016wide,covington2016deep,okura2017embedding,zhou2018deep}. 
% There are two major steps in sponsored listing ranking: estimation and ranking. For estimation, different machine learning models have been used across online marketplaces to estimate the conversion rate, for example, YouTube \cite{covington2016deep}, Google Play \cite{cheng2016wide}, and Yahoo! News\cite{okura2017embedding}. 
% We contribute to a wide range of literature focusing on designing online recommendation systems \cite{ricci2010introduction,aggarwal2016recommender,zhang2019deep}. As mentioned earlier, many online platforms first focus on developing sophisticated large-scale models to learn the conversion rates for candidate items. For instance, \cite{liu2015convolutional,cheng2016wide,covington2016deep,okura2017embedding,zhou2018deep} deploys deep-learning models for improving the estimation quality, whereas these models have been successfully implemented in applications such as YouTube \cite{covington2016deep}, Google Play \cite{cheng2016wide}, and Yahoo! News\cite{okura2017embedding}.
After estimating the conversion rates, online platforms are tasked with developing a ranking model that displays candidate items to consumers in a designated order, which is the focus of this paper. A majority of previous studies focus on the score-based ranking model \citep{zehlike2022fairness,renda2003web,montague2001relevance,manmatha2001modeling}, that is, calculating a score for each item through a customized scoring function, and then selecting the displayed items by ranking the scores. While the ranking algorithms are similar among various studies, the scoring function, however, might vary, depending on the objectives of each online platform. For instance, \cite{covington2016deep} deployed a deep neural network model and used logistic regressions to generate scores for each video content on YouTube. Google Play \citep{cheng2016wide}, instead, used only the estimated conversion rate as the score for each item, therefore, aiming to recommend the most relevant items to consumers. Netflix \citep{amatriain2013big} adopted a score function that linearly combines popularity and predicted rating for each video. \cite{yang2019bid} proposes a linear programming based algorithm for online advertisement bid, without considering the ranking. As we will introduce more details in \cref{sec:eBay display}, Marketplace A also utilizes a score-based ranking method whose objective combines different sources of revenue \citep{ge2020relevance}. In the case of ``cold start'', i.e., the launch of a new advertisement/item, where the data for conversation rate estimation is limited, \cite{ye2023cold} proposed a bandit algorithm that significantly increased the success rate of cold start in field experiments. 

An inevitable drawback of the score-based method is that even though online platforms aim to achieve Pareto efficiency over multi-objective \citep{lin2019pareto}, it is hard to guarantee the performance of individual objectives. An alternative approach is to formulate the problem as a mathematical programming problem, that is, maximizing an objective function while setting constraints on some other metrics and resources. Such a problem can be formulated as a mixed integer programming (MIP) problem \citep{celis2017ranking}, and sometimes referred to as an advertisement allocation problem \citep{mehta2013online}. We can also associate this problem with some well-studied literature on assignment problem \citep{ross1975branch} and knapsack problem \citep{salkin1975knapsack}. However, as the number of constraints increases, the MIP also becomes increasingly harder to solve, requiring various approximation techniques. A more comprehensive review can be referred to \citep{cattrysse1992survey,pentico2007assignment,karp1990optimal}. Particularly, \cite{celis2017ranking} considered the ranking problem with fairness constraints. Based on LP relaxation, they recovered integral solutions with small constraint violations. \cite{asadpour2023sequential} studied a ranking model that aims to maximize consumer engagement for online retail, proposing an algorithm by LP relaxation and randomized rounding. Further, through LP-relaxation and primal-dual pairs, Google \citep{bhalgat2012online} addressed an advertisement allocation problem with nested packing constraints, and then approximated it with an LP-based primal-dual algorithm; Microsoft \citep{chen2011real} proposed a bidding algorithm for advertisement allocation with budget constraints. Inspired by \citep{chen2011real}, Alibaba \citep{zhang2018whole}) developed a similar algorithm for a global advertisement allocation problem that optimized over accumulated search queries. 

It is worth noting that although standard algorithms for advertisement allocation problems are well established in the existing literature, newly formulated problems with distinct objectives or constraints often require tailored adaptations of these paradigms. Our proposed algorithm falls into this category. Specifically, we develop an efficient approximate algorithm to solve the underlying mixed integer programming (MIP) formulation, with direct application to the SLR problem. Furthermore, we present the first field study that empirically compares the performance of score-based and LP-based algorithms.

In the context of recommender systems, the primary design objective is to ensure low latency, that is, the ability to generate rankings within a strictly limited response time. Consequently, unlike assortment optimization problems (which we will discuss shortly), ranking problems typically do not treat item interactions as a central concern. Most studies therefore assume that position weights (often referred to as “position bias”) are independent of the item assortment \citep{chen2023bias, agarwal2019general, hofmann2013reusing} and develop methods such as propensity-score-based estimation to correct for this bias \citep{chen2023bias, hofmann2013reusing, joachims2017unbiased}.\label{s3:assortmentdifference}

\paragraph{Multi-objectives selections and trade-off}
Our paper also relates closely to the literature that studies the selection of objectives and trade-offs between muti-objectives. As mentioned earlier, given the variety of objectives, different online platforms may prioritize them differently. Other than maximizing revenue, online marketplaces might also need to consider extra objectives such as consumer utility, ranking relevance \citep{ghose2014examining}, algorithm fairness \citep{geyik2019fairness}, and content diversity \citep{chen2024effects}. For instance, online video platforms such as YouTube or Netflix emphasized more on ranking relevance \citep{covington2016deep,amatriain2013big}, while online marketplaces such as Tabao cared more about revenue \citep{hu2018reinforcement}. Naturally, each online marketplace has multiple objectives incorporated into the optimization framework. \cite{agarwal2011click} jointly optimized for initial clicks and post-click downstream. \cite{xiao2017fairness} considered the fairness conflict between individuals and the entire group. \cite{chen2024effects} recommended that social media platforms should focus on increasing content diversification specifically for active consumers. Contrary to industry assumptions, their research revealed that diversified recommendations would enhance content consumption diversity only among active consumers. \cite{adomavicius2014optimization} optimized for the aggregated diversity of a group of recommendations. \cite{sumida2021revenue} characterized the trade-off between revenue and utility under MNL, by constructing the objective function as a weighted sum of revenue and utility.  \cite{louca2019joint} proposed a similar MNL model optimization framework, weighting the sum of revenue and relevance. \cite{li2021online} considered a fairness metric evaluated by the click-throughs of different Ads and customer types. They proposed a two-stage allocation method for online advertisement, which balanced revenue and algorithm fairness. However, since the objectives might not always align, \cite{rodriguez2012multiple} addressed this issue using constraint optimization. On the other hand, as multi-objective optimization often resulted in many nondominated solutions and low pressure towards the Pareto front, \cite{li2016stochastic} developed a stochastic ranking algorithm to tackle this problem. 

\paragraph{Assortment problem}
The work is also related to the literature studying assortment optimization problems, which refer to the challenge of selecting a subset of products to offer from a larger set, in an online setting, where the goal is often to maximize revenue or customer satisfaction \citep{kok2009assortment,caro2007dynamic,gallego2020approximation}.
Differing from ranking problems, assortment problems focus on studying the interaction among listed items and do not usually study the effect of orders. In assortment literature, many previous works adopt parametric models to model the purchase probability of items given a certain set of items, such as the \textit{independent demand model} (IDM) and \textit{basic attrition model} (BAM) \citep{gallego2019revenue}. Particularly, IDM assumes that the probability of sale for each item is independent of others, thus the optimal assortment would be the one including all items with positive revenue. However, IDM might be pessimistic in a way that when items are running out of stock, recapturing any of those sales to alternate products would be infeasible since the demands are treated as independent \citep{gallego2015general}. Under BAM, on the other hand, for a given assortment, the choice probability of each item is normalized by the summation of the utility of all items in that assortment. The often-used \textit{multinomial choice model} (MNL) is a special case of BAM \citep{mcfadden1973conditional,luce2005individual}. Under the MNL model with no constraint, \cite{talluri2004revenue,gallego2004managing} characterized the revenue-ordered property of the optimal assortment. Further, for the robust assortment optimization under MNL, the optimal assortment is still revenue-ordered \citep{rusmevichientong2012robust}. However, in contrast to IDM, BAM might be optimistic about recapturing the sale of alternative items due to the self-normalization procedure \citep{gallego2015general}. \cite{feldman2022customer} also emperically showed that the MNL model performed worse during multi-purchase events. Recently, there has been some research studying the ranking of the assortment of products for online retailing. \cite{ferreira2022learning} studied the ranking of the selected assortment display to consumers. The proposed online learning algorithms also for maximum consumer engagement, which balanced between learning speed and optimality gap. As classic choice models for online retailing only gave limited consideration to substitution patterns, \cite{jiang2020high} proposed a high-dimensional choice as an alternative. \cite{asadpour2023sequential} focused on ranking assortment in the form of a submodular maximization problem. They developed offline approximation algorithms that give the maximum probability of purchases. Most of the previous works in the assortment literature focused on developing theory guarantees. The only empirical study we found on online assortment problems is \cite{feldman2022customer}. They novelly implemented the algorithm proposed in \cite{rusmevichientong2010dynamic} and presented an empirical study at Alibaba comparing the multinomial logit model versus industry-level machine learning models for assortment. 

\vspace{0.2cm}
\section{Institutional Context and Score-Based Ranking}\label{sec:eBay display}
We begin this section by outlining the relevant institutional details of our partner marketplace, including how sponsored listings rankings and the advertising system operate in Marketplace A. We then describe the current score-based ranking algorithm employed by Marketplace A.

\subsection{Sponsored Listings Rankings (SLR)}\label{subsec:market_overview}
Sponsored product advertisements are a major source of revenue for most online marketplaces, including Amazon, Alibaba, and our partner marketplace. These sponsored products typically appear in two key locations where they are relevant as a \textit{ranking problem}, which we define as generating a carefully ordered list that includes both sponsored and non-sponsored products.

The first location is the keyword search results page, where the ranking process is referred as \textit{Keyword Ranking} (KR) in our partner marketplace. KR addresses the challenge of displaying and ranking items based on consumer-entered search queries. For example, when a consumer searches for a keyword on Marketplace A, a ranked list of relevant items is presented, as illustrated in Figure \ref{fig:Search Listings display}. The task of determining the order of these listings constitutes the KR problem\label{s3:KR}.

The second common location for sponsored products is the product detail page, which consumers reach after clicking on a specific product. In this context, sponsored products may appear in sections such as "Sponsored Related Items," "Products Related to This Item," or "You May Also Like," with specific terminology varying across platforms. The corresponding display format in Marketplace A is shown in Figure \ref{fig:Sponsored listings display}. The ranking of items in this section, which includes both sponsored and non-sponsored products, is referred to as \textit{Sponsored Listings Ranking} (SLR) in our partner marketplace and is the primary focus of our study. The main product featured on the page is typically referred to as the \textit{seed item}. As with other marketplaces, Marketplace A notifies users that the listings may include sponsored items, but does not disclose which specific items are sponsored\label{s3:SLR}.

\begin{figure}
  \begin{subfigure}[c]{.5\linewidth}
    \centering
    \includegraphics[width=\linewidth]{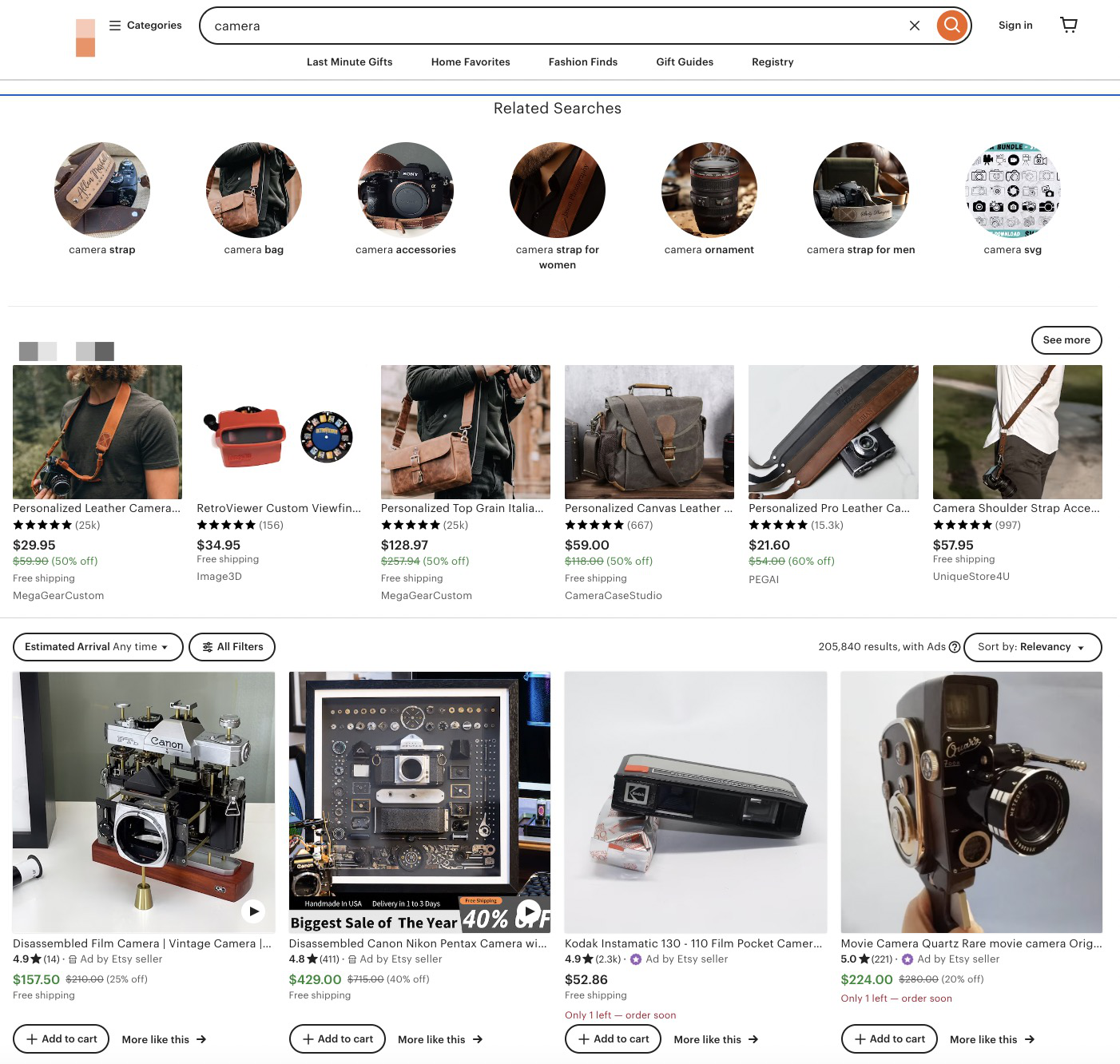}%
    \caption
      {%
        Keyword Ranking%
        \label{fig:Search Listings display}%
      }%
  \end{subfigure}\hfill
  %\begin{tabular}[c]{@{}c@{}}
    \begin{subfigure}[c]{.45\linewidth}
      \centering
      \includegraphics[width=\linewidth]{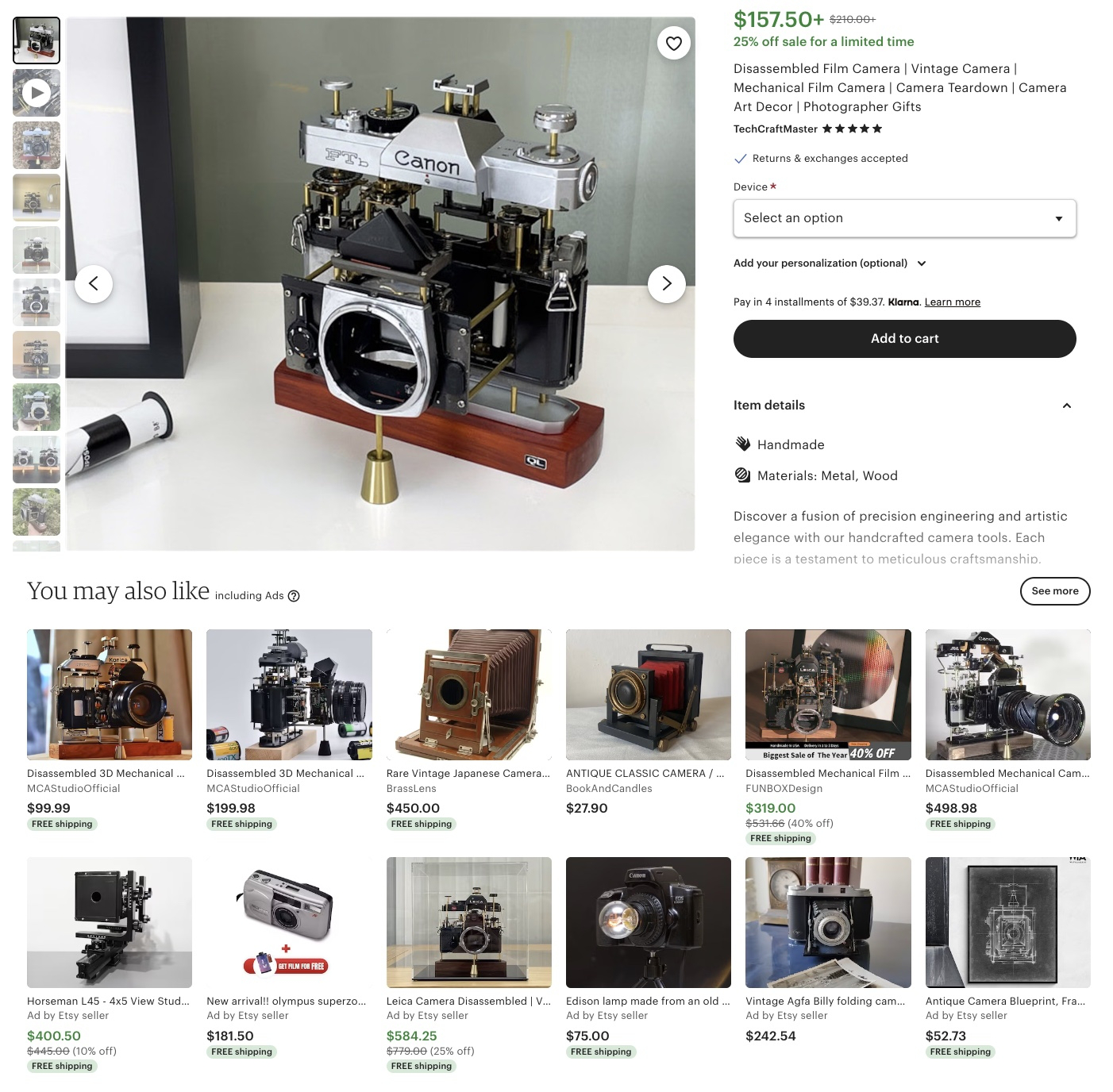}%
      \caption
        {%
          Sponsored Listings Ranking%
          \label{fig:Sponsored listings display}%
        }%
    \end{subfigure}
    %\\
  \caption
    {%
      Example of Two Ranking Problems in Marketplace A%
      \label{fig:display settings}%
    }
\end{figure}

In Marketplace A, as in many similar online marketplaces, the goal of SLR is twofold: to display relevant items to consumers and to generate substantial revenue. Achieving this dual objective is a central challenge not only for Marketplace A but also for many of its industry peers. On the one hand, Marketplace A aims to maximize revenue, with the revenue structure detailed in Section \ref{subsec:advertising_system}. On the other hand, the platform places high importance on maintaining a positive consumer shopping experience, which is crucial for its long-term growth. To support this, more relevant items, defined as those closely related to the \textit{seed item}, should be ranked higher in the sponsored listings. A detailed explanation of how relevance is measured is provided in Section \ref{subsec:score_based}. Balancing this revenue-relevance tradeoff is the core task of the ranking phase, particularly under the predict-then-optimize framework adopted by Marketplace A and widely used across the industry.

Beyond practical constraints, especially the stringent latency requirement, three inherent characteristics of the SLR problem make ranking a particularly critical and challenging task, similar to those encountered in recommender systems. First, the number of sponsored items displayed is not fixed. As discussed in Section \ref{subsec:score_based}, it can range from zero to the total number of available slots, entirely determined by the ranking algorithm. Second, there are no fixed display positions reserved for sponsored items; their placement is dynamically assigned based on the ranking algorithm’s output. Third, from the consumer’s perspective, sponsored items are not marked with any prominent label, making them indistinguishable from organic listings. These features contrast sharply with traditional sponsored ads in keyword search engines like Google. In those systems, ranking is relatively less central: a small number of fixed, clearly labeled ad slots are shown, and platforms can rely on simple boosting mechanisms for selected sponsored items without needing to rank all results across pages. As we will discuss in Section \ref{sec:conclusion}, our KR setup also differs from that of Google, which allows our proposed ranking algorithm to be extended to the KR setting in our context\label{s3:SLRfeatures}.

In Marketplace A, KR and SLR are managed by separate teams within distinct organizational units that operate independently, a common practice in large-scale marketplaces. As a result, a product that appeared in the KR results but was not clicked by consumers may still appear in the SLR section. There is currently no mechanism in place to prevent this duplication, highlighting a broader challenge faced by large organizations that rely on algorithmic optimization across siloed teams. All empirical results reported in Section \ref{sec:exp_result} are based on the SLR. In addition, we provide empirical evidence on how changes to the SLR can influence outcomes in the KR in Section \ref{subsec:counfouding}\label{s3:KRSLR}.

\subsection{Revenue Structure and Advertising System}\label{subsec:advertising_system}

Revenue from sponsored listings comprises two components: (1) the \textit{take fee} (also referred to as \textit{organic revenue}) and (2) the \textit{sponsorship fee} (also referred to as \textit{ad revenue}). For every item sold on the platform, Marketplace A collects a commission, called the \textit{take fee}, which is a fixed percentage of the item’s final sale price. This percentage, known as the \textit{take rate}, is tiered and ranges from 3\% to 15\%, depending on the item's category and price.

In addition, sellers may choose to voluntarily sponsor their items by paying an additional fee to boost visibility. This fee is calculated as a percentage of the final sale price and is referred to as the \textit{ad rate}. Marketplace A employs a bidding system for setting the \textit{ad rate}: sellers specify a bid as a percentage of the sale price when setting up a sponsored product, with allowable bids ranging from 1\% to 100\%. While sellers have the flexibility to choose any bid within the allowable range, the effectiveness of their sponsored product (i.e., its visibility and ranking) depends on how competitive their bid is relative to the trending rate\label{s3:advertising}.

Importantly, Marketplace A adopts a cost-per-acquisition (CPA) model, in which the \textit{sponsorship fee} is charged only when a product is purchased through a sponsored listing. While our discussion focuses on the SLR setting, it is worth noting that the same revenue structure, including the \textit{take rate} and \textit{ad rate}, also applies to the KR setting. 

We make two clarifying notes. First, the revenue structure, comprising both \textit{organic revenue} and \textit{advertising revenue}, is widely adopted by major marketplaces such as Amazon and Alibaba. Accordingly, the fundamental revenue–relevance tension inherent in the SLR problem is broadly applicable. Second, under this revenue structure, consumers are not constrained to purchase only one item from the sponsored listing. The final goal is to maximize revenue while maintaining relevance and quality of displayed items to sustain long-term consumer engagement\label{s3:notoneitem}. 

\subsection{Score-Based Ranking}\label{subsec:score_based}
Score-based ranking refers to ordering listed items according to the monotonic sequence of a certain score defined by the marketplace. Owing to its simplicity and computational efficiency, score-based ranking is one of the most widely adopted ranking algorithms in both recommender systems and sponsored listings operating under the predict-then-optimize paradigm. The core of this approach lies in how the score is defined. At Marketplace A, score-based ranking is the production model used for sponsored listings, and the engineering team has conducted extensive research to refine the score definition for improved performance. In this section, we present the high-level concept of the score-based ranking algorithm implemented at Marketplace A, which serves as the baseline for comparison with our proposed LP-based algorithm in the field study. Along the way, we also outline the predict-then-optimize process employed by Marketplace A and identify the specific step at which our LP-based algorithm intervenes.

When a consumer opens a webpage with a \textit{seed item}, Marketplace A first selects a collection of qualified candidate items. From now on, we use the set $\mathcal{N} = \{1,2,...,n\}$ to denote this collection. The number of candidate items $n$ ranges from 1 to 500, as some categories in Marketplace A are well-stocked with many similar items, while others in the long tail have far fewer. Meanwhile, to avoid the overwhelming number of candidate items, Marketplace A restricts the cardinality of $\mathcal{N}$ during the selection phase. Then, Marketplace A is tasked to (1) choose a subset of $\mathcal{N}$ according to the slot number $m (m \leq n)$ and then (2) display this subset of items in ranked order. The slot number $m$ typically ranges from 1 to 50. Consistent with the problem setup and the stringent latency requirements, Marketplace A implemented its SLR procedure through three main stages: (1) candidate selection, (2) purchase-through rate estimation, and (3) ranking, as detailed below.\label{s3:mnnumber}

\paragraph{Step 1: candidate selection.} 
Among all the items listed on Marketplace A, there is only a small number of items that are related to the \textit{seed item} and thus are potential candidates for the SLR. The first step is to identify candidate set $\mathcal{N}$. Here, Marketplace A takes two separate approaches. One is based on the categories of \textit{seed item}, and the other is developed from consumers' history. Each listed item has been pre-labeled with various categories and sub-categories, spanning multiple layers. The first group of candidate items is determined by their categorical similarity to the \textit{seed item}. For the second group, Marketplace A trains the \textit{k-nearest neighbors} (KNN) algorithm featured on each consumer's browsing and purchase history to select the related items. Eventually, the union of the first and the second group will formulate the candidate set. As observed in the candidate selection stage, the sponsored items can be either complements to the \textit{seed item} (selected based on the user’s browsing history) or substitutes for the \textit{seed item} (selected from the same product category)\label{s3:complementsubstitute}.

\paragraph{Step 2: purchase-through rate (PTR) and position weights estimation.}
Aligned with common practice in recommender systems, Marketplace A separates the estimation of the conversion rate, namely, the purchase-through rate (PTR), denoted as $r_j \in \mathbb{R}$, and the position weights, denoted as $h_i \in \mathbb{R}$, from the ranking problem due to latency considerations. In other words, Marketplace A follows the predict-then-optimize paradigm. For the same practical reasons, and consistent with many online marketplaces, a key underlying assumption in solving the SLR problem is that the probability that a consumer clicks on item $j$ at position $i$ can be decomposed as $r_j h_i$, thereby treating the position weights as independent of the assortment of items. In our study, we take all these estimates as given and intervene only at the ranking stage, which we introduce shortly\label{s3:PTRhiestimation}.

Because Marketplace A adopts a cost-per-acquisition (CPA) model, the key performance metric is the PTR $r_j $. Specifically, Marketplace A first generates an offline PTR estimate for each item $j \in \mathcal{N}$, which is then refined using a fine-tuned XGBoost model. This adjustment leverages consumers’ browsing and purchase histories, as well as characteristics of the \textit{seed item} and other candidate items in $\mathcal{N}$, including category, price, and sales history. Importantly, $r_j$ may vary across impressions for the same item $j$, depending on the specific consumer, the \textit{seed item}, and the other items in the candidate set. An impression refers to an instance in which a user visits a webpage displaying a designated \textit{seed item} along with a sponsored listing. If the user clicks on a product within the sponsored listing and then returns to the original \textit{seed item} page, the sponsored listing remains unchanged, and all such activities are considered part of the same impression. In our partner marketplace, PTR also serves as a proxy for the relevance of a sponsored item to the \textit{seed item}, as it captures the organic purchase likelihood conditional on the \textit{seed item}\footnote{This relevance measure is widely used by our partner marketplace and closely resembles the quality score in Google Ads, which assesses the relevance of an ad and its landing page to a user’s search query: \url{https://support.google.com/google-ads/answer/6167118?hl=en}
. Decomposing the purchase path in a sponsored listing, the final PTR reflects three main components: the expected click-through rate (CTR), the relevance of the item to the user’s latent purchase intent, and the quality of the landing product page. A key assumption in using PTR as a relevance score is that a click on the \textit{seed item} reflects the user’s latent purchase intention\label{s3:ptrqualityscore}.}. The PTR estimation process also accounts for both the relationship with the \textit{seed item} and potential cannibalization among items within the sponsored listing. This design reflects a deliberate balance between the latency constraints of real-time operations and the underlying economic relationships among items\label{s3:PTRestimation}.

Another important estimation component is the position weights $h_i$, which capture the slot-specific probability that a purchase originates from position $i$, independent of the item displayed. These weights are empirically estimated through randomized shuffling and A/B testing. In our empirical analysis, we focus on the desktop website channel, where position weights are fixed across impressions. As expected, the top slots exhibit higher weights than lower ones, and the weights drop sharply once the slot moves to the next page\label{s3:positionweights}.

\textbf{Step 3: score-based ranking.} 
In the final step, Marketplace A computes a score for each item, ranks the items based on these scores, and selects those with the highest scores for display. Consequently, not all items in the candidate set appear in the sponsored listing. The number of displayed items is not fixed but determined by Marketplace A’s internal rules. As a result, some sponsored listings may contain no sponsored products, while others may include many.

More specifically, the score $s_j$ for the $j$-th item is defined as a weighted sum of the expected \textit{organic revenue} (PTR $\times$ \textit{take rate} $\times$ price) and the expected \textit{ad revenue} (PTR $\times$ \textit{ad rate} $\times$ price) with the following form: 
\begin{equation}
    s_j = r_j \times t_j \times p_j + w \times r_j \times ad_j \times p_j,
    \label{eq:ranking score expected revenue}
\end{equation}
where $p_j$ denotes the price of the item. The term $t_j$ denotes the \textit{take rate} of the item, which is the commission levied by Marketplace A upon the sale of the item, with a predetermined rate for each item category within the marketplace. The \textit{organic revenue} reflects the mandatory platform fee paid by the marketplace for each transaction. The term $ad_j \in \mathbb R_{> 0}$ denotes the \textit{ad rate} for the $j$-th item, namely, the amount a vendor is willing to pay for increased item visibility. The process for determining $ad_j$ is detailed in Section \ref{subsec:advertising_system}. The \textit{ad revenue}, in contrast to \textit{organic revenue}, is a discretionary vendor expense aimed at boosting product prominence. 

The parameter $w \in \mathbb{R}_{+}$ governs the relative importance of \textit{organic revenue} versus \textit{ad revenue}. It is optimized at the impression level based on factors including the predicted PTR, position weights, ad rates, and item prices. The method for optimizing $w$ is detailed in \cite{ge2020relevance}.  Notably, the selection of the scoring function and the calibration of $w$ has gone through extreme tuning at Marketplace over an extended period, and we believe their production model optimally harnesses the power of the score-based algorithm\label{s3:westimation}. 

Next, we will propose an alternative ranking approach based on \textit{linear programming} (LP) to solve the SLR. Note that the first two steps, candidate selection and PTR and position weights estimation, remain unchanged. Our contribution lies in introducing a new method for the final ranking stage.

\section{Linear Programming (LP)-Based Ranking Algorithm}\label{sec:LP-model}
In this section, we first formulate the SLR problem as a \textit{mixed integer programming} (MIP) model. However, solving this MIP exactly is computationally infeasible under the platform’s latency constraints. We therefore propose a primal–dual algorithm to solve its LP relaxation, which produces a feasible and near-optimal solution to the original MIP. We then describe the production algorithm implemented by Marketplace A, a practical variant of our primal–dual algorithm, that enforces the relevance constraints in expectation rather than at each impression, and present its associated theoretical guarantee. Finally, we provide a theoretical explanation for why our LP-based approach can outperform the conventional score-based algorithm.

\subsection{Mixed-Integer Programming (MIP) Formulation for the SLR Problem}\label{subsec:MIP}
Suppose we have $n$ candidate items designated to be assigned to $m$ slots $(m \leq n)$. For each slot $i$, we have the knowledge of its estimated position weight $h_i$, which represents the importance of the $i$-th slot. In general, the customer is more likely to buy the items at the first few slots than the latter ones, and such differences among slots are measured by $h_i$. Additionally, for each item $j$, we take the estimates: (1) the relevance to the \textit{seed item}, $r_j$, measured by the PTR, and (2) expected revenue $v_j$ as the sum of expected \textit{organic revenue} and \textit{ad revenue} as we defined previously from equation \eqref{eq:ranking score expected revenue}, i.e., $v_j := r_j p_j ( t_j + ad_j)$. Denote the vectors as $\boldsymbol{v} := (v_1,\cdots,v_n), \boldsymbol{r} := (r_1,\cdots,r_n), \boldsymbol{h} := (h_1,\cdots,h_m)$.

Without loss of generality, we assume that position weights $h_i$ are sorted in descending order; namely, the item at the slot with a smaller index is more likely to be purchased than the same item at a slot with a larger index.

\begin{ass}[Monotonically decreasing position weights]\label{ass:decreasing position weights}
    The position weights $h_i$ is sorted in descending order, i.e., $h_{i} \geq h_{i+1}$.
\end{ass}

In the MIP formulation, we aim to maximize the expected total revenue while posing a lower bound on the overall relevance level. The binary decision variable $\boldsymbol{X} \subseteq \{0, 1\}^{m \times n}$ represents a ranking plan, namely, an item $j$ is assigned to slot $i$ when $\boldsymbol{X}_{ij}$ takes the value $1$, or $0$ if not. As a feasible ranking plan satisfies $\sum_{i = 1}^{m}\boldsymbol{X}_{ij}\le 1$ (i.e., each slot can only have one candidate), $\sum_{j = 1}^{n}\boldsymbol{X}_{ij}\le 1$ (i.e., each candidate can only be shown at most once), and $\boldsymbol{X}_{ij}\in \{0,1\}$ (i.e., the decision variable is binary). For notational convenience, we denote
\begin{equation}
    \mathcal{P} 
    = 
    \left\{\boldsymbol{X}: \sum_{i=1}^{m} \boldsymbol{X}_{ij} \leq 1 \forall j, \sum_{j=1}^{n} \boldsymbol{X}_{ij} \leq 1 \forall i, \boldsymbol{X}_{ij} \in \{0,1\}\right\}, \label{eq:MIP X const 2}
\end{equation}
as the feasible region of a ranking. Furthermore, we can compute the maximal overall relevance as
\begin{align}\label{eq:maximumrelevance}
    a = \max_{\boldsymbol{X} \in \mathcal{P} }
    \left\{
    \sum_{i=1}^{m} \sum_{j=1}^{n} h_i r_j \boldsymbol{X}_{ij}
    \right\}.
\end{align}
$a$ is the maximal possible relevance a ranking can achieve. That is, when the ranking is based solely on the relevance (PTR) of an item to the \textit{seed item}, without accounting for revenue from sponsored products, the maximum achievable relevance score corresponds to the highest possible ranking outcome\label{s3:maximumrelevance}. In the MIP formulation, we aim to find a ranking that maximizes the revenue while making sure that the overall relevance is at least $\lambda$ times the maximal overall relevance with a pre-determined parameter $\lambda \in [0,1]$:

\begin{align}\label{problem:MIP}
    (P_0) \quad \quad \max_{\boldsymbol{X}} 
    & \sum_{i=1}^{m} \sum_{j=1}^{n} h_i v_j \boldsymbol{X}_{ij}\\
    \text{s.t. } 
    & \sum_{i=1}^{m} \sum_{j=1}^{n} h_i r_j \boldsymbol{X}_{ij} - \lambda a \geq 0 \label{eq:MIP lambda const}\\
    &\boldsymbol{X} \in \mathcal{P}\label{eq:permutation domain}\ .
\end{align}

Notice that the MIP formulation is always feasible, because the ranking that gives the maximal overall relevance store is a feasible solution.
Additionally, with $\lambda$ decreases, the feasible region of \hyperref[problem:MIP]{$P_0$} is monotonically enlarged. Since a smaller $\lambda$ will induce a possibly lower relevance value but higher revenue, the trade-off between revenue and relevance is fulfilled by the parameter $\lambda$. 

We highlight several key assumptions in our model setup. First, we assume that the position weight $h_i$ is independent of the PTR. While this independence may not hold universally across all sponsored listings, the assumption is motivated by the need for low-latency response times. This simplification is commonly adopted in industry practice \citep{chen2023bias}. Second, while our ranking stage does not explicitly account for interactions among items in a sponsored listing due to latency considerations, we address potential dependencies through PTR estimation. That is, the estimated PTR for each item can depend on the overall assortment of items in the sponsored listing. However, the position weight $h_i$ remains fixed and is not influenced by the assortment. This setup reflects a deliberate design choice: we aim to capture item interactions during prediction stage rather than the optimization stage. Third, our model allows for the possibility that a consumer purchases multiple items from a single sponsored listing. This flexibility is consistent with platform objectives, as multiple purchases do not negatively impact the platform, sellers, or consumers. In practice, our partner marketplace data show that, among sponsored listings that led to purchases, most involved only one item, and the maximum number of items purchased from a single listing was four.\label{s3:modelassumption}

\subsection{Linear Programming (LP) Relaxation and an Efficient Algorithm}
A challenge of the real-time SLR problem is the latency requirement, which needs to be finished within 0.1 second in general. This poses a strict limitation for obtaining a ranking policy, preventing us from exactly solving MIP. Indeed, the traditional Branch-and-Bound/Branch-and-Cut method for solving MIP costs excessively long and unstable time that the marketplace cannot afford. In such cases, a typical approach is to look at and solve its LP relaxation. Unfortunately, the state-of-the-art commercial LP solvers can still not solve the LP relaxation within the latency requirement (see Appendix \ref{sec:Computational results for LP-based algorithm} for a computational study). Here, we present a linear time dual-based bisection algorithm that finds a feasible and approximately optimal MIP solution.

% In this section, we present a dual-based simple algorithm that can approximately solve its LP relaxation, which leads to an approximate solution to the MIP formulation. 

% show that, through leveraging the powerful tools of LP relaxation and the primal-dual approach, we are able to develop approximation algorithms with high performance, which eventually gives near-optimal solutions to original MIP problems with very low time complexity. 

We start with presenting the LP relaxed version of MIP, by relaxing the binary variable constraint $\bfX_{ij}\in\{0, 1\}$ to $\bfX_{ij}\in[0, 1]$:
\begin{align}\label{problem: OPT_LP}
    (P_1) 
    \quad \quad \max_{\boldsymbol{X}} 
    & \sum_{i=1}^{m} \sum_{j=1}^{n} h_i v_j \boldsymbol{X}_{ij}\\
    s.t. &\; \sum_{i=1}^{m} \sum_{j=1}^{n} h_i r_j \boldsymbol{X}_{ij} - \lambda a \geq 0 \label{eq:LP lambda const}\\
    &\boldsymbol{X} \in \mathcal{F} \label{eq:LP feasible region},
\end{align}
where
\begin{align}
    \mathcal{F} 
    = 
    \left\{\boldsymbol{X}: \sum_{i=1}^{m} \boldsymbol{X}_{ij} \leq 1 \forall j, \sum_{j=1}^{n} \boldsymbol{X}_{ij} \leq 1 \forall i \label{eq:LP X const 2}, \boldsymbol{X}_{ij} \in [0,1]\right\} \ .
\end{align}
% with the optimal value of \hyperref[problem: OPT_LP]{$P_1$} denoted as $\text{OPT}_{\text{LP}}$. 
%The objective function is denoted as $e: X \mapsto \sum_{i=1}^{m} \sum_{j=1}^{n} h_i v_j \boldsymbol{X}_{ij}$. Further, the left hand side function of constraint \eqref{eq:LP lambda const} is denoted as $c: X \mapsto \sum_{i=1}^{m} \sum_{j=1}^{n} h_i r_j \boldsymbol{X}_{ij} - \lambda a$.

Algorithm \ref{al:find_opt_x} presents our major algorithm. The basic idea of the algorithm is to do a bisection search for the optimal dual variable to the constraint \eqref{eq:LP lambda const}, and then obtain the corresponding primal response that gives an approximate solution to the LP formulation. More specifically, the algorithm has the following three steps:

\begin{algorithm}
    \SetAlgoLined
    {\bf Input:} $\boldsymbol{\delta}, m$.\\
    {\bf Output:} $\boldsymbol{X} \in \mathcal{F} \subseteq \mathbb{R}^{m\times n}$.\\
    Sort the value of $\boldsymbol{\delta}$ in the descending order, i.e., $\delta_{j_1} \ge \cdots \ge \delta_{j_n}$. \footnotemark\\
    Compute $\boldsymbol{X}$, such that $\boldsymbol{X}_{1 j_1} = \boldsymbol{X}_{2 j_2} = \cdots = \boldsymbol{X}_{m j_m}=1$, and $\boldsymbol{X}_{ij}=0$ elsewhere.  \\
    \textbf{Return} $\boldsymbol{X}$
    \caption{PrimalResponse($\boldsymbol{\delta}, m$)}
    \label{al:primal_response}
\end{algorithm}

\footnotetext{{In the case of a tie, we can break the tie by lexicographical order. This ensures a unique output of the function $\text{PrimalResponse}$ for any given input. }}

\textbf{Step 1: Check if the relevance constraint is redundant}\label{subsubsec:Cases when the relevance constraint is redundant}
\\
Given $\boldsymbol{h}, \boldsymbol{v}, \boldsymbol{r}$, we could first easily verify whether the relevance constraint \eqref{eq:MIP lambda const} is redundant. Removing \eqref{eq:MIP lambda const}, the optimal value of \hyperref[problem:MIP]{$P_0$} would be $\max_{X \in \mathcal{F}} \{\sum_{i=1}^{m} \sum_{j=1}^{n} h_i v_j \boldsymbol{X}_{ij}\}$. Similar to the way of obtaining $a$, due to the unimodularity property, we can obtain the optimal solution by (1) first sorting the vector $\boldsymbol{v}$ in descending order, i.e., $v_{j_1} \geq \cdots \geq v_{j_n}$, and (2) set $\boldsymbol{X}_{ij_i} = \cdots = \boldsymbol{X}_{mj_m} = 1$ and $\boldsymbol{X}_{ij} = 0$ elsewhere. Also, we can see that $\boldsymbol{X} \in \mathcal{P}$. 

If $\sum_{i=1}^{m} \sum_{j=1}^{n} h_i r_j \boldsymbol{X}_{ij} \geq \lambda a$, it means that the ranking that generates the maximum revenue is also feasible to \hyperref[problem:MIP]{$P_0$}, indicating that the relevance constraint \eqref{eq:LP lambda const} is redundant and $X$ is an optimal solution to \hyperref[problem:MIP]{$P_0$}. Then, the algorithm ends, and an optimal solution to the LP and MIP formulation is found.  
%Since $\bar{X} \in \mathcal{P}$, then it is also feasible to \hyperref[problem:MIP]{$P_0$}. Thus, $\bar{X}$ is an optimal solution to \hyperref[problem:MIP]{$P_0$}, where $e(\bar{X}) = \text{OPT}_{\text{LP}} = \text{OPT}_{\text{MIP}}$.

\textbf{Step 2: Obtain a range of the optimal dual variable}

Let us consider the following primal-dual form of \hyperref[problem: OPT_LP]{$P_1$}, denoted as \hyperref[problem: OPT_LP_Dual]{$D_1$}:
\begin{align}\label{problem: OPT_LP_Dual}
    (D_1) \quad \quad \min_{\mu\ge 0} D(\mu):= \max_{\boldsymbol{X}\in \mathcal{F}} 
    & \sum_{i=1}^{m} \sum_{j=1}^{n} h_i \boldsymbol{X}_{ij} (v_j + \mu r_j) - \mu \lambda a\ .
\end{align}
Denote $\boldsymbol{\delta} := (\delta_1, \cdots, \delta_n)$, where $\delta_j=v_j + \mu r_j$. 
%Then we define the function $f: \mu, X \mapsto \sum_{i=1}^{m} \sum_{j=1}^{n} h_i \boldsymbol{X}_{ij} \delta_j$. 
We denote the optimal value of \hyperref[problem: OPT_LP_Dual]{$D_1$} as $\text{OPT}_{\text{D}}$, and an optimal dual variable is $\mu^*$.

Suppose that constraint \eqref{eq:LP lambda const} is not redundant, then it follows from complementary slackness that $\mu^*>0$. The second step of Algorithm \ref{al:find_opt_x} is to find an upper and lower bound on $\mu^*$.

\setlength{\intextsep}{4pt} % use this to reduce the space around algorithm

\begin{algorithm}
	\SetAlgoLined
	{\bf Input:} $\boldsymbol{h}, \boldsymbol{r}, \boldsymbol{v},\lambda,a, \epsilon$, $m$.\\
 
        \Comment{\textbf{Step 1}: check if the relevance constraint \eqref{eq:MIP lambda const} is redundant.} \\
        Compute $\boldsymbol{X} = \text{PrimalResponse}(\boldsymbol{v}, m)$
        . \\
        \uIf{$\sum_{i=1}^{m} \sum_{j=1}^{n} h_i r_j \boldsymbol{X}_{ij} - \lambda a \geq 0$}{{\bf return} $\boldsymbol{X}$}
        \Else{
        \Comment{\textbf{Step 2}: find the range containing $\mu^*$.} \\
        $\mu_-=0, \mu_+=1$\\
        Compute $\bfdelta = \bfv + \mu_+ \bfr$. \\
        Compute $\boldsymbol{X} = \text{PrimalResponse}(\bfdelta, m)$.\\
        \While{$\sum_{i=1}^{m} \sum_{j=1}^{n} h_i r_j \boldsymbol{X}_{ij} - \lambda a <0$}{$\mu_-=\mu_+$, $\mu_+=2\mu_+$\\
        Compute $\bfdelta = \bfv + \mu_+ \bfr$. \\
        Compute $x=\text{PrimalResponse}(\bfdelta, m)$.\\}

        \Comment{\textbf{Step 3}: find $\mu^*$.} \\
        \While{$ \mu_{+} - \mu_{-} > \epsilon $}{
                $\mu = \frac{\mu_{+} + \mu_{-}}{2}$ \\
                Compute $\bfdelta = \bfv + \mu \bfr$. \\
                Compute $\boldsymbol{X} = \text{PrimalResponse}(\bfdelta, m)$.\\
                \uIf{$\sum_{i=1}^{m} \sum_{j=1}^{n} h_i r_j \boldsymbol{X}_{ij} - \lambda a \ge 0$}{
                $\mu_{+} = \mu$ \\
                }\Else{
                    $\mu_{-} = \mu$.
                    % \uIf{$c(\bar{X}(\mu)) < 0$}{
                    %     $\mu_{-} = \mu$
                    % }
                    % \Else{
                    %     \textbf{return} $\bar{X}(\mu)$ 
                    %     \Comment{$\bar{X}(\mu)$ is the optimal solution.}
                    % }
            }
    	}
        Compute $\bfdelta = \bfv + \mu_{+} \bfr$. \\
        Compute $\boldsymbol{X} = \text{PrimalResponse}(\bfdelta, m)$.\\
        {\bf return} $\boldsymbol{X}$
        }
	\caption{LP-based method}
	\label{al:find_opt_x}
\end{algorithm}

\setlength{\intextsep}{0pt}

Notice that for any given $\mu$, one can solve the inner maximization problem over $\bfX$ in \eqref{problem: OPT_LP_Dual} by a simple ranking-based algorithm (see Algorithm \ref{al:primal_response}), and we denote the solution $\bfX(\mu)$. Furthermore, we know that $\bfX(\mu)$ gives a sub-gradient of $D(\mu)$ by
\begin{equation}
    \sum_{i=1}^{m} \sum_{j=1}^{n} h_i r_j \boldsymbol{X}_{ij}(\mu) - \lambda a\in \partial D(\mu).
\end{equation}
Notice that $D(\mu)$ is a convex function and $0\in \partial D(\mu^*)$. Furthermore, due to the feasibility of the constraint \eqref{eq:LP lambda const}, we know there exists a large enough finite $\mu$, such that its sub-gradient is strictly positive. Thus, we can obtain an upper-bound $\mu_+\ge \mu^*$ by iteratively doubling the value of $\mu_+$ and checking its corresponding sub-gradient. We can also obtain a lower bound $\mu_-\le \mu^*$ during this process.

\textbf{Step 3: Obtain a near-optimal dual solution and a feasible ranking}
When a lower bound $\mu_-$ and an upper bound $\mu_+$ of $\mu^*$, we can do a bisection search for $\mu^*$ till $\mu_+-\mu_-\le \epsilon$, where $\epsilon$ is the target accuracy. Then it is easy to check that the primal response of $\mu_+$ is a feasible and near-optimal solution to \hyperref[problem:MIP]{$P_0$}. Eventually, Algorithm \ref{al:find_opt_x} outputs the primal response of $\mu_+$.

The only parameter in our LP-based algorithm is the relevance parameter $\lambda$. We provide guidance on how to fine-tune this parameter in Appendix \ref{subsec:guidience}.

\subsection{Production Variant of the LP-Based Algorithm and Theoretical Guarantee}\label{subsec:theoreticalguarantee}
The original algorithm described in \Cref{al:find_opt_x} guarantees strictly feasible ranked listings, that is, every impression satisfies the relevance constraint. In practice, however, Marketplace A does not require each individual impression to strictly meet the relevance constraint. Instead, it is acceptable for the constraint to hold in expectation, allowing for higher overall revenue.

To address this production need, we introduce a variant of \Cref{al:find_opt_x}, presented as \Cref{al:find_opt_mu}, which is used in our field experiment. This new algorithm retains the same primal-dual structure but modifies how the ranked listing is constructed. Specifically, in addition to the feasible solution from \Cref{al:find_opt_x}, we compute an infeasible (with a minor gap) but revenue-optimal solution using a relaxed dual variable $\mu_{-}$ (with $\boldsymbol{X} = \text{PrimalResponse}(\bfv + \mu_{-} \bfr, m)$), where the relevance score of this listing may fall slightly below the required threshold $\lambda a$.

Finally, \Cref{al:find_opt_mu} returns a randomized solution sampled from the feasible and infeasible solutions such that the expected relevance constraint is satisfied, while achieving higher expected revenue. This randomized strategy is provably both feasible in expectation and optimal under the relaxed constraint. Before we prove this, we first introduce some assumptions.

\begin{algorithm}
    \SetAlgoLined
	{\bf Input:} $\boldsymbol{h}, \boldsymbol{r}, \boldsymbol{v},\lambda,a, \epsilon$, $m$.\\
 
        \Comment{\textbf{Step 1}: check if the relevance constraint \eqref{eq:MIP lambda const} is redundant.} \\
        Compute $\boldsymbol{X} = \text{PrimalResponse}(\boldsymbol{v}, m)$
        . \\
        \uIf{$\sum_{i=1}^{m} \sum_{j=1}^{n} h_i r_j \boldsymbol{X}_{ij} - \lambda a \geq 0$}{{\bf return} $\boldsymbol{X}$}
        \Else{
        \Comment{\textbf{Step 2}: find the range containing $\mu^*$.} \\
        $\mu_-=0, \mu_+=1$\\
        Compute $\bfdelta = \bfv + \mu_+ \bfr$. \\
        Compute $\boldsymbol{X} = \text{PrimalResponse}(\bfdelta, m)$.\\
        \While{$\sum_{i=1}^{m} \sum_{j=1}^{n} h_i r_j \boldsymbol{X}_{ij} - \lambda a <0$}{$\mu_-=\mu_+$, $\mu_+=2\mu_+$\\
        Compute $\bfdelta = \bfv + \mu_+ \bfr$. \\
        Compute $x=\text{PrimalResponse}(\bfdelta, m)$.\\}

        \Comment{\textbf{Step 3}: find $\mu^*$.} \\
        \While{$ \mu_{+} - \mu_{-} > \epsilon $}{
                $\mu = \frac{\mu_{+} + \mu_{-}}{2}$ \\
                Compute $\bfdelta = \bfv + \mu \bfr$. \\
                Compute $\boldsymbol{X} = \text{PrimalResponse}(\bfdelta, m)$.\\
                \uIf{$\sum_{i=1}^{m} \sum_{j=1}^{n} h_i r_j \boldsymbol{X}_{ij} - \lambda a \ge 0$}{
                $\mu_{+} = \mu$ \\
                }\Else{
                    $\mu_{-} = \mu$.
                    % \uIf{$c(\bar{X}(\mu)) < 0$}{
                    %     $\mu_{-} = \mu$
                    % }
                    % \Else{
                    %     \textbf{return} $\bar{X}(\mu)$ 
                    %     \Comment{$\bar{X}(\mu)$ is the optimal solution.}
                    % }
            }
    	}
        Compute $\bfdelta_{+} = \bfv + \mu_{+} \bfr$. \\
        Compute $\bfdelta_{-} = \bfv + \mu_{-} \bfr$. \\
        Compute $\boldsymbol{X}_{+} = \text{PrimalResponse}(\bfdelta_{+}, m)$.\\
        Compute $\boldsymbol{X}_{-} = \text{PrimalResponse}(\bfdelta_{-}, m)$.\\
        %{\bf return} $\boldsymbol{X}$
        
    	%$\mu = \frac{\mu_{+} + \mu_{-}}{2}$ \\
        \Comment{\textbf{Step 4}: randomized return.} \\
        
        $U = \text{Uniform}(0,1)$.\\
        $\alpha = (\sum_{i=1}^{m} \sum_{j=1}^{n} h_i r_j \boldsymbol{X}_{+, ij} - \lambda a)/(\sum_{i=1}^{m} \sum_{j=1}^{n} h_i r_j \boldsymbol{X}_{+, ij} - \sum_{i=1}^{m} \sum_{j=1}^{n} h_i r_j \boldsymbol{X}_{-, ij})$. \\
    	\uIf{$ U > \alpha$}{
            \textbf{return} $\boldsymbol{X}_{+}$ \\
        }\Else{
            \textbf{return} $\boldsymbol{X}_{-}$ \\
        }
    	%\textbf{return} $\mu_{-}$, $\mu_{+}$, $\bar{x}(\mu_{-})$, and $\bar{x}(\mu_{+})$.
        }
	\caption{A variant of LP-based algorithm in production}
	\label{al:find_opt_mu}
\end{algorithm}

First, we assume that in SLR setting, there exists no degeneracy: 
% \begin{ass}[Unique Ranked Order]\label{ass:unique_rank_order}
%     For any given vector $\boldsymbol{v}, \boldsymbol{r}$ in our SLR problem \hyperref[problem:MIP]{$P_0$}, then for the component-wise division between $\boldsymbol{v}, \boldsymbol{r}$, defined as $(\boldsymbol{v}/\boldsymbol{r}) := (v_1/r_1, \cdots, v_n/r_n)$, we assume that there is no tie of elements. That is, there are no two different indices $j_1, j_2 \in [n]$ such that $v_{j_1}/r_{j_1} = v_{j_2}/r_{j_2}$. Note that we previously assume that $r_j > 0, \; \forall \; j \in [n]$ to avoid the infinity of $v_j/r_j$. 
% \end{ass}
\begin{ass}[Nondegeneracy]\label{ass:unique_rank_order}
    For an index set $\mathcal{I} := \{(j,k): j \in [m] \text{ or } k \in [m], j \neq k, r_k \neq r_j\}$, we define $\mu_{jk} := \frac{v_j - v_k}{r_k - r_j}$ for all pairs $(j, k) \in \mathcal{I}$. Then for any given vector $\boldsymbol{v}, \boldsymbol{r}$ in our SLR problem \hyperref[problem:MIP]{$P_0$}, and a pair of indicies$(j,k) \in \mathcal{I}$, we define the pairwise indifference value
    \begin{equation}
        \mu_{jk} := \frac{v_j - v_k}{r_k - r_j},
    \end{equation}
    and its collection as $\{\mu_{jk}\}_{(j,k) \in \mathcal{I}}$. Then for any two distinct unordered pairs $(j,k) \neq (p,q)$ with $(j,k), (p,q) \in \mathcal{I}$, we assume that $\mu_{jk} \neq \mu_{pq}$. It is noted that $\mu_{jk}$ is a value such that there might be some ``ties'' in the vector of $\boldsymbol{v} + \mu_{jk} \boldsymbol{r}$, i.e., $v_j + \mu_{jk} r_{j} = v_k + \mu_{jk} r_{k}$.
\end{ass}
%there exists no degeneracy at the associated optimal dual variable $\mu^*$. In other words, there exists at most two indices $i, j \in [n]$ such that ${v}_{i} + \mu^* {r}_{i} = {v}_{j} + \mu^* {r}_{j}$. 
Notice that both values $\bfv$ and relevance $\bfr$ are learned from a machine learning model, thus one can view they come from a continuous distribution. As a direct consequence, it is with probability 1 that nondegeneracy holds.

% Then, we have another assumption upper bounding the error term $\epsilon$ in \Cref{al:find_opt_mu}:
% \begin{ass}[Upper bound of $\epsilon$]\label{ass:small_error}
%     \lz{Nonuniform version:}
%     For any dual variable $\mu \in \mathbb R_{\geq 0}$, let $\pi(\mu) := \text{argsort}(\boldsymbol{v} + \mu \boldsymbol{r})$ denote the permutation that sorts the elements of $\boldsymbol{v} + \mu \boldsymbol{r}$ in ascending order. Then we assume that in \Cref{al:find_opt_mu}, $\epsilon$ is small enough such that $\pi(\mu)$ is unchanged for all $\mu \in (\mu_{*}, \mu_{+})$ and $(\mu_{-}, \mu_{*})$, respectively.  In other words, for all $\mu_{1}, \mu_{2} \in (\mu_{*}, \mu_{+})$, $\pi(\mu_1) = \pi(\mu_2)$, and the same for all $\mu_{1}, \mu_{2} \in (\mu_{-}, \mu_{*})$, meaning that there are no other pairwise indifference values in between $(\mu_{-}, \mu_{+})$, except for $\mu^*$.\\
%     \lz{Uniform version:} With \Cref{ass:unique_rank_order}, we define the smallest distance between all pairs of indifference values as 
%     \begin{equation}
%         \Delta_{\min}(\boldsymbol{r}, \boldsymbol{v}) 
%         :=
%         \min_{(j,k), (p, q) \in \mathcal{I}, (j,k) \neq (p, q)}
%         \left|\mu_{jk} - \mu_{pq}\right|.
%     \end{equation}
%     Then we assume that $\epsilon < \Delta_{\min}(\boldsymbol{r}, \boldsymbol{v})$.
% \end{ass}

\begin{thm}[Optimality of the approximated solution of \Cref{al:find_opt_mu}]\label{thm:optimality_lp_algo}
    Under \Cref{ass:decreasing position weights,ass:unique_rank_order}, we define the smallest distance between all pairs of indifference values as 
    \begin{equation}
        \Delta_{\min}(\boldsymbol{r}, \boldsymbol{v}) 
        :=
        \min_{(j,k), (p, q) \in \mathcal{I}, (j,k) \neq (p, q)}
        \left|\mu_{jk} - \mu_{pq}\right|.
    \end{equation}
    Let $\Tilde{\boldsymbol{X}}$ be the randomized solution returned by \Cref{al:find_opt_mu} with the bisection tolerance $\epsilon < \Delta_{\min}(\boldsymbol{r}, \boldsymbol{v})$, and $\boldsymbol{X}^{\text{LP}}$ be the optimal solution to the relaxed primal problem \hyperref[problem: OPT_LP]{$P_1$}. Then, in expectation,  $\Tilde{\boldsymbol{X}}$ equals to $\boldsymbol{X}^{\text{LP}}$, namely, the output solution in expectation is the optimal solution to the LP: 
        \begin{equation}
            \mathbb{E}[\Tilde{\boldsymbol{X}}] = \boldsymbol{X}^{\text{LP}}, \quad \mathbb{E}[\text{Rel}(\Tilde{\boldsymbol{X}})] = \lambda a, \quad 
            \mathbb{E}[\text{Rev}(\Tilde{\boldsymbol{X}})] = \text{OPT}_{\text{LP}}. 
        \end{equation}
\end{thm}

Proof of Theorem~\ref{thm:optimality_lp_algo} is provided in Appendix~\ref{subsec:theorem}.

\begin{remark}[Runtime Complexity of \Cref{al:find_opt_x,al:find_opt_mu}] 
    Define $\mu_{\max} := \max_{(j,k) \in \mathcal{I}} \{\mu_{jk}\}$. Then the runtime complexity of \Cref{al:find_opt_x,al:find_opt_mu} are the same and both depends on three dominant components: (1) sorting and selection from \Cref{al:primal_response} with $O(n \log m)$ per each step, (2) step 2 (finding the range containing $\mu^*$) with $O(\log (1 + \mu_{\max}))$, and (3) step 3 (find $\mu^*$) with $O(\log(\mu_{\max}/\epsilon))$. Thus the overall runtime complexity is $O((\log(1 + \mu_{\max}) + \log (\mu_{\max}/\epsilon)) \cdot n \log m) = O(n \log(\mu_{\max}/\epsilon)\log m)$, which is a nearly-linear-time algorithm in $n$.
\end{remark}
\subsection{Why the LP-Based Algorithm Outperforms the Score-Based Algorithm}\label{subsec:whytheoretical}
The central challenge in our setting is the revenue–relevance tradeoff. If advertising revenue were absent, the marketplace’s objective would be straightforward: promote items consumers are most likely to purchase, thereby maximizing \textit{organic revenue}. However, because advertising contributes substantially to marketplace revenue, prioritizing sponsored items does not necessarily align with consumer preferences. The tradeoff, therefore, reflects a fundamental tension: displaying more sponsored items increases marketplace revenue but may reduce relevance if those items are less likely to be purchased.

In the score-based algorithm, the industry convention is to treat this as a tradeoff between \textit{organic revenue} and \textit{ad revenue}. As expressed in Equation~\ref{eq:ranking score expected revenue}, the balancing is achieved through the parameter $w$, which adjusts the weight assigned to ad-related factors. In contrast, our LP-based algorithm adopts a different perspective and modeling approach. Rather than distinguishing revenue sources, we treat the marketplace’s objective as maximizing total revenue (Equation~\ref{problem:MIP}), regardless of whether it arises from organic purchases or advertisements. We incorporate the relevance dimension explicitly by imposing a minimum requirement, controlled by parameter $\lambda$. This framing captures the idea that marketplaces often set hard constraints on consumer experience, ensuring a minimum relevance threshold, while still aiming to maximize overall revenue. Crucially, our approach also delivers an LP solution that satisfies latency requirements, enabling practical deployment. Section~\ref{subsec:mechanism} will provide empirical evidence on why this approach can outperform the score-based algorithm and illustrate how the revenue–relevance tradeoff plays out in practice\label{s3:mechansimtheory}.

Several clarifications are in order. First, both our proposed algorithm and the score-based algorithm incorporate slot-specific weights $h_i$. As described in Section~\ref{subsec:score_based}, the parameter $w$ in the score-based method is optimized at the impression level, accounting for predicted PTR, position weights, ad rates, and item prices (see \citealt{ge2020relevance} for details)\footnote{Our modeling idea could, in principle, be expressed in a score-based framework as $\delta_j = r_j \times t_j \times p_j + r_j \times ad_j \times p_j + \mu \times r_j$, where $\mu$ is the dual variable associated with the relevance constraint. We nonetheless prefer the LP-based formulation because it naturally accommodates hard constraints on relevance, consistent with marketplace goals of safeguarding consumer experience and long-term performance, while still maximizing revenue subject to those constraints. Moreover, the LP-based formulation possesses the modeling flexibility to accommodate additional constraints and can be effectively applied in an online setting\label{s3:generalscore}.}. Second, in designing our field experiment, we deliberately adopted a conservative test strategy. Specifically, while $w$ in the score-based algorithm is optimized at the impression level, our $\lambda$ is applied globally across all impressions. To ensure a fair and conservative benchmark, we set $\lambda$ above the average relevance level achieved by the score-based algorithm in pre-experiment data. This design allows us to test whether our LP-based approach can improve both revenue and relevance through a more principled formulation of the tradeoff\label{s3:experimentclarification}.

\section{Field Experiment Design and Data}\label{sec:experiment design and data}
To showcase the effectiveness of the LP-based algorithm, we conduct a randomized field experiment in Marketplace A to evaluate its empirical performance. This section first outlines the experiment design, followed by a description of the data used in our analysis.

\subsection{Experiment design}\label{subsec:experiment design}
We conducted a 19-day randomized field experiment from April 27 to May 16, 2022, to evaluate the empirical performance of the proposed LP-based algorithm on Marketplace A's platform across four major countries. Marketplace A operates four primary consumer access channels: desktop website, iOS app, Android app, and mobile website. The experiment was implemented on the desktop website channel ("dWeb"), which generates the highest consumer traffic among all channels. As discussed in Section \ref{subsec:score_based}, because our experiment focuses on a single channel, the slot-specific weights $h_i$ keep the same across all impressions in both the treatment and control groups. The benchmark was the score-based algorithm described in Section \ref{subsec:score_based}, which served as Marketplace A’s production model during the study period. Introduced in 2020, the score-based algorithm had undergone extensive experimentation and hyperparameter tuning, and the version deployed in our experiment is considered to have reached its full optimization potential.

We considered two treatment groups for the proposed LP-based algorithm, differing in the relevance parameter $\lambda$. Following the parameter tuning approach described in Section \ref{subsec:guidience}, we selected $\lambda\in\{0.9, 0.95\}$ for the field study. 
Both values exceed the average relevance scores achieved by the score-based algorithm in the pre-experiment data. This choice reflects a conservative test design aimed at examining whether the LP-based algorithm can deliver revenue gains while simultaneously improving the relevance of sponsored listings. Specifically, we implemented the following three experiment conditions:
\begin{enumerate}
    \item \textbf{LP90 (Treatment 1):} For consumers assigned to this group, the displayed listings will be generated by the LP-based algorithm with hyperparameter $\lambda = 0.90$.
    \item \textbf{LP95 (Treatment 2):} For consumers assigned to this group, the displayed listings will be generated by the LP-based algorithm with hyperparameter $\lambda = 0.95$.
    \item \textbf{Benchmark (Control):} For consumers assigned to this group, the displayed listings will be generated by the score-based algorithm.
\end{enumerate}

During the 19-day experiment, each consumer visiting the sponsored listings page was randomly assigned to one of the three groups, with listings generated according to the assigned method. Assignment was based on unique consumer IDs, ensuring that each consumer consistently received the same treatment throughout the experiment, regardless of the number of visits.

\subsection{Data}\label{subsec:data}
During the 19-day period, our field experiment involved in total 32,058,321 consumers, of which  10,686,419 consumers were assigned to \textbf{LP90 (Treatment 1)}, 10,689,158 consumers were assigned to \textbf{LP95 (Treatment 2)}, and 10,682,744 were assigned to \textbf{Benchmark (Control)}. During the experiment period, it is likely that each consumer pays multiple visits, generating many impressions. As a result, in total, 329 million impressions were generated. Specifically, 109,868,303 impressions were assigned to \textbf{LP90 (Treatment 1)}, 109,843,211 impressions were assigned to \textbf{LP95 (Treatment 2)}, and 109,567,246 impressions were assigned to \textbf{Benchmark (Control)}. 

The randomization check results are presented in Table~\ref{tab:Stat of randomization check}. We examine three key variables, calculated using data from the three months preceding the experiment launch: (1) the average number of purchases per consumer (\textit{Purchases}); (2) the total expenditure (USD) per consumer (\textit{GMV}); and (3) the total revenue generated for the partner marketplace (USD) per consumer (\textit{Revenue}). Due to confidentiality agreement, exact values cannot be disclosed; instead, we normalize the benchmark to one and report all other values relative to this benchmark. We also provide the minimum pairwise p-value across all group comparisons. The results confirm that the randomization was successful\label{s3:randomization}.

\vspace{0.2cm}
\begin{center}
\captionsetup{labelfont=bf}
\captionof{table}{Randomization Check.} \label{tab:Stat of randomization check} 
\scalebox{1.0}{
    \begin{tabular}{cccccc}
    \hline \hline
    & LP90 & LP95 & Benchmark & Mininum Pairwise P-value\\
    & (Treatment 1) & (Treatment 2) & (Control) & \\
    \hline
    %\multicolumn{5}{c}{Panel A: Summary Statistics for experiment data} \\
    %\hline
    %Number of consumers & 10,686,419 & 10,689,158 & 10,682,744  & >0.99\\
    %Number of impressions & 24,457,371 & 24,435,422 & 24,404,213 & >0.94\\
    %Number of impressions & 109,868,303 & 109,843,211 & 109,567,246 &>0.94 \\
    %Total GMV(USD) & 5,935,196 & 5,840,514 & 5,761,382 & -  \\
    %Total GMV(USD) & * & * & * & -  \\
    %Total revenue(USD)  & 312,756 & 315,734 & 310,060 & - \\
    %Total revenue(USD)  & * & * & *0 & - \\
    %\hline 
    %\multicolumn{5}{c}{Panel B: Randomization Check for pre-experiment data (per consumer)} \\
    %\hline
    % Expenditure(USD) & 5.28 & 5.33 & 5.21 &>0.22\\
    % Revenue(USD) &0.43 &0.42 &0.43 &>0.94\\
    % Number of purchases &0.11 &0.11 &0.11 &>0.99 \\
    Purchases & 1.00 & 1.00 & 1.00 & 0.99 \\
    GMV & 1.01 & 1.02 & 1.00 & 0.22 \\ 
    %Expenditure(USD) & * & * & * &0.22\\
    Revenue & 1.00 & 0.98 & 1.00 &0.94\\
    %Number of purchases & * & * & * &0.99 \\
    Observations & 10,686,419 & 10,689,158 & 10,682,744 & - \\
    \hline \hline
    \end{tabular}
}
%\begin{tablenotes}
 %     \small
    %  \item \textit{Notes.} * Due to the confidentiality of data sources, consumer-level data is forbidden to be disclosed by Marketplace A. Nevertheless, we are able to demonstrate the pairwise p-value to show the randomness of the experiment assignment.
%\end{tablenotes}
\end{center}
\vspace{0.2cm}

The dataset provided to us consists of three components. The first comprises aggregated results from the full-scale experiment, including: (1) the average number of purchases per consumer for each group (\textit{Purchases}); (2) the total expenditure (USD) per consumer for each group (\textit{GMV}); (3) the total revenue generated for the partner marketplace (USD) per consumer for each group (\textit{Revenue}); and (4) the average transaction price per consumer among those who made at least one purchase during the experiment for each group (\textit{Transaction Price}). All of these measures only document the activities from the SLR. We use this data to evaluate the effectiveness of our proposed LP-based algorithm in Section~\ref{subsec:main_result}\label{s3:datafirst}.

Because of the vast scale of the full dataset and practical restrictions related to data privacy, the partner marketplace was unable to provide us with impression-level data for the entire experiment. Instead, we were given a randomly sampled subset of the full-scale experiment, comprising 690,611 impressions per experiment group and 2,071,833 impressions in total. Recall that an impression refers to an instance in which a user visits a webpage displaying a designated \textit{seed item} along with a sponsored listing. For each impression, we observe several variables: the average product price in the sponsored listing, the average normalized PTR weighted by position weights, the total number of clicks, the number of purchases, the corresponding revenue generated for the partner marketplace, the date of the impression, and the experiment group assignment. Summary statistics for this sample are reported in Appendix \ref{subsec:sampledata}. Notice that the consumer ID of each impression was not provided to us for this sample. Although the data are randomly sampled, there remain some performance differences between the sampled dataset and the full dataset, which we will further discuss in Section~\ref{sec:exp_result}. We primarily use this dataset to analyze the mechanisms through which our proposed LP-based algorithm achieves superior performance. As we will show, the core advantage of our modeling approach in balancing the revenue–relevance tradeoff remains robust.\label{s3:onlineimpression}

%the transaction price (if a purchase occurred), the ranking position of the transaction product within the sponsored listing, 
To complement our mechanism analysis, the third dataset consists of a random sample drawn from the 17 days preceding the launch of the field experiments, containing 266,003 impressions and 13,370,159 unique items. Unlike the random sample obtained during the experiments, this dataset provides richer detail at the item level. For each impression, we can observe the price, the \textit{take rate}, the \textit{ad rate}, the normalized PTR, and the ranking for each item in the candidate set. In addition, the data documents whether an item appeared in the final sponsored listing, whether a consumer clicked on or purchased a specific item in the listing, as well as the price and revenue of the \textit{seed item}, the seller ID, and the consumer ID. A distinction is that, whereas the field experiment data only covers registered users with user IDs (which enable randomization), this pre-experiment sample may include impressions without user IDs. Specifically, 164,737 impressions are from 160,694 registered users, while the remaining 101,266 impressions are from unregistered users. Summary statistics for this dataset are provided in Appendix \ref{subsec:sampledata}.

\section{Field Experiment Analysis and Results}\label{sec:exp_result}
This section reports our empirical analysis and principal findings. We first assess the effectiveness of the proposed LP-based algorithm using the key outcome variables defined by the partner marketplace (Section~\ref{subsec:main_result}). We then examine the mechanisms underlying its superior performance relative to the score-based algorithm (Section~\ref{subsec:mechanism}). Finally, we consider additional factors pertinent to the generalizability and limitations of the experiment results (Section~\ref{subsec:counfouding}).

\subsection{Does LP-Based Algorithm OutPerform the Score-Based Algorithm?}\label{subsec:main_result}
We evaluate algorithmic effectiveness in the partner marketplace using three outcome variables, as defined in Section~\ref{subsec:data}: (1) the average number of purchases per consumer (\textit{Purchases}); (2) total expenditure per consumer (\textit{GMV}); and (3) total revenue generated for the partner marketplace per consumer (\textit{Revenue}). The main results are reported in Table~\ref{tab:Stat financial performance}. Due to confidentiality agreement, we report only the relative changes of the two treatment groups (\textbf{LP90} and \textbf{LP95}) under the proposed LP-based algorithm compared with the control group (\textbf{Benchmark}), i.e., the score-based algorithm. All monetary values are calculated in USD. We also present the corresponding p-values from pairwise t-tests.\label{s3:outcomemeasure}

\vspace{0.2cm}
\begin{table}[ht]
\centering
\begin{threeparttable}
\caption{LP-Based Algorithm vs Score-Based Algorithm.} %\vspace{.2in}
\begin{tabular}{ccc}
    \hline \hline
    & LP90 & LP95 \\
    & (Treatment 1) & (Treatment 2) \\
    \hline
    %Revenue per consumer (USD) & 0.0293 &  0.0295  & 0.0290 \\
    \textit{Purchases} & 1.55\%*** \quad (<0.001) & 3.22\%*** \quad (<0.001)  \\
    %\\
    \textit{GMV} & 1.39\% \quad (0.11)  & 2.96\%*** \quad (<0.001) \\
    %\\
    \textit{Revenue} & 1.80\%** \quad (0.03) & 0.81\% \quad (0.32)  \\
    %GMV per consumer (USD) & 0.5553 & 0.5465 & 0.5393\\
    
    %Purchases per consumer & 0.0162 & 0.0159 & 0.0157  \\
    
    %Average transaction price (USD)  & 34.36 & 34.37 & 34.44 \\
    %\textit{Average transaction price (USD)} & -0.23\% & -0.20\% \\
    \hline \hline
    \end{tabular}
\begin{tablenotes}
\footnotesize
\item Notes. We report the relative changes of the two treatment groups (\textbf{LP90} and \textbf{LP95}) compared with control group (\textbf{Benchmark}), i.e., the score-based algorithm. Associated $p$-values from pairwise $t$-tests are reported in parentheses. The sample size in Treatment 1 is $N_{\text{treatment1}} = 10{,}686{,}419$, in Treatment 2 is $N_{\text{treatment2}} = 10{,}689{,}158$, and in the Control group is $N_{\text{control}} = 10{,}682{,}744$. *** $p<$ 0.01, ** $p<$ 0.05, * $p<$ 0.1.
\end{tablenotes}
\label{tab:Stat financial performance}
\end{threeparttable}
\end{table}
\vspace{0.2cm}  

\begin{comment}
  \vspace{0.2cm}
\begin{center}
\captionsetup{labelfont=bf}
\captionof{table}{LP-Based Algorithm vs Score-Based Algorithm} \label{tab:Stat financial performance} 
\scalebox{1.0}{
    \begin{tabular}{ccc}
    \hline
    & LP95 & LP90 \\
    & (Treatment 1) & (Treatment 2) \\
    \hline
    %Revenue per consumer (USD) & 0.0293 &  0.0295  & 0.0290 \\
    \textit{Purchases} & 3.22\%*** \quad (<0.001) & 1.55\%*** \quad (<0.001) \\
    \textit{GMV} & 2.96\%*** \quad (<0.001) & 1.39\% \quad (0.11) \\
    \textit{Revenue} & 0.81\% \quad (0.32) & 1.80\%** \quad (0.03) \\
    %GMV per consumer (USD) & 0.5553 & 0.5465 & 0.5393\\
    
    %Purchases per consumer & 0.0162 & 0.0159 & 0.0157  \\
    
    %Average transaction price (USD)  & 34.36 & 34.37 & 34.44 \\
    %\textit{Average transaction price (USD)} & -0.23\% & -0.20\% \\
    \hline 
    \end{tabular}
}
 \begin{tablenotes}
      \small
      \item \textit{Notes.} We report the relative changes of the two treatment groups (\textbf{LP95} and \textbf{LP90}) compared with control group (\textbf{Benchmark}), i.e., the score-based algorithm. Associated $p$-values from pairwise $t$-tests are reported in parentheses. The sample size in Treatment 1 is $N_{\text{treatment1}} = 10{,}686{,}419$, in Treatment 2 is $N_{\text{treatment2}} = 10{,}689{,}158$, and in the Control group is $N_{\text{control}} = 10{,}682{,}744$. *** $p<$ 0.01, ** $p<$ 0.05, * $p<$ 0.1.
      \end{tablenotes}
\end{center}
\vspace{0.2cm}  
\end{comment}

First, in terms of consumer purchases, the \textbf{LP95} treatment leads to a statistically significant increase of 3.22\% ($p<0.001$), while the \textbf{LP90} treatment yields a smaller but still significant increase of 1.55\% ($p<0.001$). This suggests that the LP-based algorithm is more effective at stimulating demand and generating additional consumer transactions than the score-based algorithm.\label{s3:outcomeresult1}

Second, with respect to GMV, the \textbf{LP95} treatment yields a statistically significant increase of 2.96\% ($p<0.001$), whereas the \textbf{LP90} treatment shows only a modest increase of 1.55\%, which is statistically insignificant. Recall that the relevance parameter $\lambda$ is higher in \textbf{LP95}. The stronger performance of \textbf{LP95} relative to \textbf{LP90} highlights the importance of a sufficiently stringent relevance threshold in driving consumer spending compared with the score-based algorithm. Importantly, the improvements in GMV are not attributable to changes in transaction prices: the relative changes in average transaction price are small and statistically insignificant for both treatments (\textbf{LP95}: –0.23\%; \textbf{LP90}: –0.20\%). This pattern suggests that the gains in GMV arise primarily from higher purchase volumes rather than higher prices, and that such volume-driven growth materializes only when the relevance threshold is sufficiently high. 

Finally, with respect to marketplace revenue, \textbf{LP95} produces only a modest and statistically insignificant gain of 0.81\% ($p=0.32$), whereas \textbf{LP90} generates a statistically significant increase of 1.80\% ($p=0.03$). Taken together with the GMV results, this suggests that higher consumer spending does not necessarily translate into higher marketplace revenue, in part because revenue also depends on advertising income. These patterns point to a relevance–revenue tradeoff: while a higher relevance threshold (\textbf{LP95}) enhances consumer spending, it does not substantially increase marketplace revenue, whereas a lower relevance threshold (\textbf{LP90}) has limited impact on GMV but raises revenue through greater advertising returns. It is important to note that in our experiment design, both treatment groups have higher average relevance scores than the benchmark score-based algorithm (see Section~\ref{subsec:experiment design}). Thus, the \textbf{LP90} results demonstrate the superiority of the LP-based algorithm in improving marketplace revenue while simultaneously enhancing the relevance of sponsored listings.

\subsection{Why the LP-Based Algorithm Outperforms: Empirical Insights}\label{subsec:mechanism}
In this section, we present empirical evidence to explain why the LP-based algorithm outperforms the score-based benchmark. We begin with a heterogeneity analysis, showing that listings facing a stronger revenue–relevance tradeoff benefit more from our method relative to the score-based algorithm. We then analyze the characteristics of items prioritized by the LP-based algorithm, offering deeper insight into the mechanisms underlying its superior performance. Because this analysis requires fine-grained item-level information, we rely on the richer impression-level data collected before the experiment’s launch, and we complement it with aggregated listing-level evidence from the sampled impression-level data in the field experiment. Finally, drawing on these results together with the theoretical arguments in Section~\ref{subsec:whytheoretical}, we synthesize the evidence to highlight why the LP-based algorithm outperforms the score-based benchmark\label{s3:mechanismlogic}.

\subsubsection{Sponsored Listings that Benefit More}\label{subsubsec:heterogeneity}
In this section, we conduct a heterogeneous treatment effect analysis to provide empirical evidence for why our LP-based algorithm outperforms. Using the sampled impression-level data from the full experiment, we focus on how the algorithm addresses the core revenue–relevance tradeoff. The strength of this tradeoff varies across contexts, and our framework should perform particularly well when it is more pronounced. One natural dimension along which this occurs is the displayed price of sponsored listings. For higher-priced sponsored listings, the forces shaping conversion probability and revenue potential diverge more sharply: on the one hand, expensive products typically face lower baseline purchase probabilities, a pattern we also confirm in our data; on the other hand, each successful conversion yields disproportionately higher revenue for the platform. Thus, effectively managing the revenue–relevance tradeoff becomes especially critical for higher-priced listings, and it is precisely in such cases where our LP-based approach is expected to deliver the greater gains\label{s3:HTEargument}. It is worth noting that the \textbf{LP90} group is the one that can illustrate this argument, since it is the group that achieves higher revenue than the benchmark based on Table~\ref{tab:Stat financial performance}. As discussed in Section~\ref{subsec:data}, our sampled dataset exhibits performance differences compared to the full dataset. The key observation is that while the \textbf{LP90} group did not show a significantly higher number of purchases than the score-based benchmark, it still achieved higher overall revenue. Detailed results are provided in Appendix~\ref{subsec:representation}.

To evaluate our hypothesis, we estimate the following empirical specification:
\small
\begin{align*}
    Y_{i} 
    & = 
    \beta_{0} + \beta_{1} \times \text{Algo90}_i   + \beta_2 \times \textit{LogAvgPrice}_{i} + \beta_3 \times \textit{WeightedPTR}_{i} 
  + \beta_4 \times \text{Algo90}_i \times \textit{LogAvgPrice}_{i}
     + \eta_t + \epsilon_{i},
\end{align*}
\normalsize
where for each impression $i$, the dependent variable $Y_{i} \in \{$\textit{Purchases}, \textit{Revenue}$\}$. $\textit{Purchases}$ measures the number of purchases per impression, and $\textit{Revenue}$ denotes the revenue generated for the partner marketplace per impression. $\text{Algo90}_i$ is the treatment indicator capturing assignment to the \textbf{LP90} with the omitted category being the \textbf{Benchmark}. $\textit{LogAvgPrice}_i$ is the logarithm of the average product price of all items in the sponsored listing within impression $i$; the log transformation reduces the influence of extreme price values. $\textit{WeightedPTR}_i$ denotes the average normalized PTR, weighted by position weights for the sponsored listing in impression $i$. We include day fixed effects $\eta_t$ to control for time-varying heterogeneity. Our main parameter of interest is $\beta_4$, which captures the interaction between algorithm assignment and displayed product price, and is expected to be positive.

\vspace{0.2cm}
\begin{table}[ht]
\centering
\begin{threeparttable}
\caption{Sponsored Listings that Benefit More.} %\vspace{.2in}
\begin{tabular}{ccc}
    \hline \hline
    & Purchases & Revenue \\
    & (1) & (2) \\
    \hline
    LP90 & -0.019*** & -0.091**\\
    & (0.002) & (0.044)\\
    %\\
    %LP95 & 0.009*** & -0.019\\
   % & (0.002) & (0.052)\\
    %\\
    LogAvgPrice & -0.076*** & -0.054***\\
    & (0.000) & (0.015) \\
    %\\
    WeightedPTR & 0.002*** & 0.212*** \\
    &  (0.000) & (0.006)\\
    %\\
    LP90 $\times$ LogAvgPrice & 0.004*** & 0.034**\\
      & (0.000) & (0.015) \\
    %\\
     %LP95 $\times$ LogAvgPrice & -0.001** & 0.009\\
    %  & (0.001) & (0.017)\\
    \hline \hline
    \end{tabular}
\begin{tablenotes}
\footnotesize
\item Notes. Observation numbers are $N = 1,381,222$. Robust standard errors are reported in parentheses. *** $p<$ 0.01, ** $p<$ 0.05, * $p<$ 0.1.
\end{tablenotes}
\label{tab:heterogeneity}
\end{threeparttable}
\end{table}
\vspace{0.2cm}  

Table~\ref{tab:heterogeneity} reports the results of the heterogeneity analysis. First, consistent with intuition, sponsored listings with a higher average displayed price tend to receive fewer purchases and generate lower platform revenue, whereas listings with a higher predicted transaction rate (PTR) attract more purchases and contribute to higher revenue. Second, the negative coefficient on \textbf{LP90} arises because the random sample does not show a significantly higher number of purchases for the \textbf{LP90} group, suggesting that for lower-priced product listings, our algorithm delivers limited incremental value. Third, the positive and significant coefficient on the interaction term $LP90 \times \log(AvgPrice)$ supports our prediction. Specifically, for sponsored listings with higher average product prices, where the revenue–relevance tradeoff is more pronounced, the LP-based algorithm demonstrates greater improvement relative to the score-based benchmark\label{s3:HTEresults}.

\subsubsection{Item Characteristics Favored by the LP-based Algorithm}\label{subsubsec:lpfavored}
To further understand why the LP-based algorithm can outperform the score-based algorithm, we next examine what types of items it tends to favor. 

Given the data limitations, our sampled impression-level data from the experiment only includes the aggregated characteristics of sponsored listings, but our exercise needs item-level information. To achieve this, we leverage the impression-level data collected prior to the experiment’s launch, which provides richer item-level details. For all data details, please refer to Section \ref{subsec:data}. 

This dataset enables us to disentangle the two key stages of ranking algorithms (LP-based versus score-based): \textit{item selection}, which determines the set of products included in the sponsored listing, and \textit{conditional ranking}, which specifies their ordering once selected. Distinguishing between these stages is important, as it allows us to assess not only whether the LP-based algorithm selects a different set of items compared to the score-based benchmark, but also whether it employs a distinct logic in ranking them. 

To achieve our goal, we simulate both the score-based algorithm and the LP-based algorithm (with $\lambda \in \{0.90, 0.95\}$) on the same candidate set for each impression, producing three distinct ranked lists with an equal number of displayed items, set to the partner marketplace’s average display size. The simulation reveals substantial divergence in \textit{item selection}. Out of 13,370,159 candidate items across 266,003 impressions, the score-based algorithm selects 2,579,511 unique items, while the LP-based algorithm selects 2,530,810 items ($\lambda = 0.90$) and 2,535,457 items ($\lambda = 0.95$). The overlap between the score-based and LP-based selections is only 1,395,573 items ($\lambda = 0.90$) and 1,419,438 items ($\lambda = 0.95$), corresponding to overlap ratios of just $37.6\%$ and $38.4\%$ of the combined unique item sets. These statistics provide early evidence that the two ranking algorithms employ markedly different logics in selecting items for sponsored listings.

We implement the following empirical strategy to evaluate how the algorithms differ in \textit{item selection}. The analysis is conducted separately for the \textbf{LP90} and \textbf{LP95} algorithms at the item level. For each comparison (\textbf{LP90} vs. score-based algorithm or \textbf{LP95} vs. score-based algorithm), we construct a universe consisting of all unique items shown by either algorithm across the dataset. The following logistic regression specifications are estimated:
\begin{align*}
     P(\text{Selected}_{j}^{LP_k} = 1) = \Lambda(\beta_{01} + \beta_{11} \text{ExpectedRevenue}_j + \beta_{21} \text{Popularity}_j ),\\
   P(\text{Selected}_{j}^{LP_k} = 1) = \Lambda(\beta_{02} + \beta_{12} \text{OrganicRevenue}_j + \beta_{22} \text{AdRevenue}_j + \beta_{32} \text{Popularity}_j ).%\\
   %P(\text{Selected}_{j}^{LP_k} = 1) = \Lambda(\beta_{03} + \beta_{13} \text{PTR}_j + \beta_{23} \text{LogPrice}_j + \beta_{33} \text{AdRate}_j + \beta_{43} \text{Popularity}_j ).
\end{align*}
where $\Lambda(\cdot)$ denotes the logistic function. For each unique item $j$, the dependent variable $\text{Selected}_{j}^{LP_k}$ equals 1 if the item is selected by the LP algorithm with constraint $\lambda = k$ ($k \in \{0.90, 0.95\}$), and 0 otherwise. $\text{ExpectedRevenue}_j$ denotes the expected revenue of item $j$, computed using Equation~\ref{eq:ranking score expected revenue} with $w=1$ and $t_j$ set to the average \textit{take rate} of our partner marketplace, since the item-level \textit{take rate} is not observed. Because the \textit{take rate} exhibits little variation in our setting, this approximation suffices to illustrate the mechanism of our method. $\text{Popularity}_j$ is defined as the total number of times item $j$ appears in the candidate set and serves as a proxy for item popularity. $\text{OrganicRevenue}_j$ corresponds to the \textit{organic revenue} component, given by $r_j \times t_j \times p_j$ in Equation~\ref{eq:ranking score expected revenue}, while $\text{AdRevenue}_j$ corresponds to the \textit{ad revenue} component, given by $r_j \times \text{ad}_j \times p_j$. %In addition, $\text{PTR}_j$ denotes the normalized PTR, $\text{LogPrice}_j$ is the logarithm of the item’s price, and $\text{AdRate}_j$ is the item’s \textit{ad rate}.

\vspace{0.2cm}
\begin{table}[ht]
\centering
\begin{threeparttable}
\caption{Item Characteristics Favored by the LP-Based Algorithm: Item Selection.} %\vspace{.2in}
\begin{tabular}{ l p{2.2cm} p{2.2cm} p{2.2cm} p{2.2cm}  }
    \hline \hline
    & \multicolumn{2}{c}{LP90}  & \multicolumn{2}{c}{LP95}  \\
    & (1) & (2) & (3) & (4) \\
    \hline
    ExpectedRevenue & 0.764*** &    & 0.207*** &  \\
    &  (0.006) &  & (0.003) &  \\
   % \\
    OrganicRevenue & &  0.148*** &  & 0.072***  \\
     &   &(0.004) &  & (0.002) \\
    %\\
    AdRevenue &  & 3.163*** & & 0.817*** \\
     &  & (0.019) &  & (0.010) \\
    %\\
    %PTR &  & & 224.462*** & & & 286.653***\\
    %&  & & (0.458) & & & (0.524)\\
    %\\
    %LogPrice & & & 0.200*** & & & 0.077*** \\
    %&  & & (0.001) & & & (0.001)\\
    %\\
    %Ad Rate &  & & 0.676*** & & & -0.638***  \\
    %&  & & (0.011) & & & (0.010)\\
    %\\
    Popularity & 0.118*** & 0.114*** & 0.135*** & 0.133***\\
    & (0.001) & (0.001)& (0.001)& (0.001)\\
    %\\
    Observations & 3,793,666 & 3,793,666  &3,793,666 & 3,793,666  \\
    \hline \hline
    \end{tabular}
\begin{tablenotes}
\footnotesize
\item Notes. Robust standard errors are reported in parentheses. *** $p<$ 0.01, ** $p<$ 0.05, * $p<$ 0.1.
\end{tablenotes}
\label{tab:selection_model}
\end{threeparttable}
\end{table}
\vspace{0.2cm}  
Results on item characteristics favored by the LP-based algorithm in the selection stage are reported in Table~\ref{tab:selection_model}. Two patterns are noteworthy. First, by more effectively capturing the revenue–relevance tradeoff, both \textbf{LP90} and \textbf{LP95} identify items with higher expected revenue, encompassing both organic and advertising revenue, relative to the score-based benchmark in the \textit{item selection} stage. Second, our LP-based algorithms favor more popular items. A potential concern is that such selection may exacerbate inequities in online marketplaces, as niche products face reduced opportunities for promotion. An advantage of our LP-based framework, however, is its flexibility: fairness constraints can be readily incorporated (e.g., \citealt{luzhu2025}), and we provide additional extensions along these lines in Appendix \ref{sec:overall planning}\label{s3:itemselectionresults}.%Second, when we further decompose organic versus advertising revenue, a distinction emerges between the two LP variants: \textbf{LP90} tends to select items with higher PTR, higher prices, and higher ad rates, whereas \textbf{LP95} selects items with higher PTR and higher prices but lower ad rates. This divergence reflects the stronger relevance constraint embedded in the \textbf{LP95} formulation. 

Next, we move to the evaluation of differences in \textit{conditional ranking}. We adopt a similar strategy with the following specification:
\begin{align*}
     P(\text{Promoted}_{ij}^{k} = 1 ) = \Lambda(\beta_{01} + \beta_{11} \text{ExpectedRevenue}_j+ \beta_{21} \text{Popularity}_j  ),\\
   P(\text{Promoted}_{ij}^{k} = 1 ) = \Lambda(\beta_{02} + \beta_{12} \text{OrganicRevenue}_j + \beta_{22} \text{AdRevenue}_j+ \beta_{32} \text{Popularity}_j ).%\\
   %P(\text{Promoted}_{ij}^{k} = 1 ) = \Lambda(\beta_{03} + \beta_{13} \text{PTR}_j + \beta_{23} \text{LogPrice}_j + \beta_{33} \text{AdRate}_j  ).
\end{align*}
This impression–item–level analysis is restricted to the subset of items appearing in both algorithms’ lists for a given impression $i$. For each such item $j$, the dependent variable $\text{Promoted}_{ij}^{k}$ equals 1 if $\text{Rank}_{ij}^{\text{LP}} < \text{Rank}_{ij}^{\text{score}}$ (i.e., the item is ranked higher by the LP algorithm) and 0 otherwise. The evaluation is again conducted separately for \textbf{LP90} and \textbf{LP95}.

\vspace{0.2cm}
\begin{table}[ht]
\centering
\begin{threeparttable}
\caption{Item Characteristics Favored by the LP-based Algorithm: Conditional Ranking.} %\vspace{.2in}
\begin{tabular}{ l p{2.4cm} p{2.4cm} p{2.4cm} p{2.4cm} }
    \hline \hline
    & \multicolumn{2}{c}{LP90}  & \multicolumn{2}{c}{LP95}  \\
    & (1) & (2) & (3) & (4) \\
    \hline
    ExpectedRevenue &0.001*** &   & 0.001*** &  \\
    &  (0.000) &  & (0.000) &  \\
    %\\
    OrganicRevenue & & 0.004***  & & 0.004***  \\
     & & (0.001) & &  (0.001) \\
    %\\
    AdRevenue & & -0.002*** &  & -0.002***   \\
     &  &(0.000) & & (0.000) \\
    %\\
    Popularity & -0.001*** & -0.001*** & -0.001*** & -0.001***\\
    & (0.000)& (0.000) & (0.000) & (0.000) \\
   %\\
    Observations & 1,574,818  & 1,574,818  &  1,602,775 & 1,602,775\\
    \hline \hline
    \end{tabular}
\begin{tablenotes}
\footnotesize
\item Notes. Robust standard errors are reported in parentheses. *** $p<$ 0.01, ** $p<$ 0.05, * $p<$ 0.1.
\end{tablenotes}
\label{tab:ranking_model}
\end{threeparttable}
\end{table}
\vspace{0.2cm}  

Results on the item characteristics favored by the LP-based algorithm during the ranking stage are reported in Table~\ref{tab:ranking_model}. Consistent with the selection stage, the algorithm continues to favor items with higher expected revenue, indicating that the LP-based approach can effectively capture the revenue–relevance tradeoff not only when deciding which items to include, but also when determining how to rank them relative to each other. A key difference emerges, however, in how the algorithm weights the components of revenue. In the ranking stage, the algorithm tends to promote items with higher organic revenue but down-weight items with higher ad revenue. This suggests that once the set of sponsored listings is determined, the algorithm emphasizes consumer-driven revenue rather than ad-driven revenue when refining the ranking. Another noteworthy finding is that the LP-based algorithm systematically de-emphasizes more popular items in the ranking stage. By less aggressively promoting popular items, the LP-based approach helps mitigate concerns about fairness and concentration of exposure across items in the selection stage\label{s3:conditionalrankingresults}.

Similar patterns can also be seen in the aggregated sponsored-listing-level evidence from the sampled impression-level data in the field experiment. We report the results in Appendix \ref{subsec:itemsupplementary}.

\subsubsection{Summarization of Why the LP-Based Algorithm Outperforms}\label{subsubsec:summarizationwhy}
The LP-based algorithm outperforms the score-based benchmark because it tackles the core challenge of balancing the revenue–relevance tradeoff in a more principled and flexible manner. Whereas the score-based approach frames this tradeoff as a weighted tension between \textit{organic revenue} and \textit{ad revenue}, the LP-based framework directly maximizes total revenue while imposing a minimum relevance constraint, thereby reflecting the practical requirement that marketplaces must safeguard consumer experience even as they pursue revenue growth. Empirically, we show that the performance gains of the LP-based algorithm are larger precisely when the revenue–relevance tradeoff is more pronounced. Moreover, the algorithm proves advantageous in both the \textit{item selection} and \textit{conditional ranking} stages, consistently identifying items with higher expected revenue. Finally, in Appendix~\ref{subsec:revenuerelevance}, we document additional evidence that illustrates how the revenue–relevance tradeoff operates in practice\label{s3:summarizationwhy}.

\subsection{Generalizability and Limitations of the Experiment Results }\label{subsec:counfouding}
Although our field experiment provides empirical evidence for the superiority of the LP-based algorithm, we acknowledge two concerns regarding the generalizability and limitations of our findings.

First, all outcome variables in our main analysis capture only activities within the SLR. A possible concern is that while purchases and revenues from the SLR increase, they could partly cannibalize the purchases and revenues of the \textit{seed item}. To assess this possibility, we analyze the sampled data collected prior to the launch of the field experiment, which includes revenue information for the \textit{seed item} (Appendix \ref{subsec:SLRKR}). Our results suggest that such cannibalization is unlikely to be substantial.

Second, our experiment focuses on short-run effects, raising the question of whether the observed treatment effects persist in the long term. On the one hand, the main advantage of the LP-based framework lies in its ability to capture the revenue–relevance trade-off, a mechanism that we believe can be sustained in ranking scenarios. Indeed, our empirical evidence in Section \ref{subsec:mechanism} supports this mechanism. On the other hand, we attempt to provide additional evidence on the potential for long-term persistence using the available data (Appendix \ref{subsec:longterm}). While these analyses are limited, they offer suggestive insights that our experiment
results are likely to hold in the long run\label{s3:longrun}.

\section{Conclusion}\label{sec:conclusion}
In this paper, we propose a novel approach to ranking items in sponsored listings that not only improves performance but also satisfies industry latency requirements. Specifically, we develop a LP–based algorithm for sponsored listings. Beyond establishing its theoretical properties, we conduct a large-scale field experiment on a major U.S. online marketplace. The results validate both the superiority of our LP-based algorithm over the industry’s score-based benchmark and its applicability in real-time settings. Complementary empirical analyses further shed light on the mechanisms driving its effectiveness. The key advantage of our framework lies in its principled treatment of the revenue–relevance tradeoff: rather than framing it as a weighted balance between \textit{organic revenue} and \textit{ad revenue}, as in the score-based approach, our algorithm directly maximizes total revenue subject to a minimum relevance constraint. We further highlight the flexibility of LP-based methods in incorporating broader planning constraints.

Several important avenues remain for future research. First, while our empirical evidence focuses on sponsored listings, the framework could be extended to KR settings where ranking plays a critical role, such as in our partner marketplace. A key distinction is that sponsored items are typically specified to consumers in KR, a challenge that would need to be addressed in such extensions. Second, following industry practice, we estimated position weights ($h_i$) and PTR separately and treated them as given when optimizing the ranking process. Developing efficient methods to jointly estimate $h_i$ and PTR represents an interesting and important next step. Finally, although we provide preliminary evidence on the long-term effects of our LP-based algorithm, richer field experiments specifically designed to capture sustained outcomes would yield deeper insights into its long-run impact (\citealt{zhanglongterm2020}).\label{s3:limitation} In particular, future investigations could examine how the estimation of PTR should be adapted to capture consumers’ behavioral changes under the new ranking algorithm, as well as how advertisers adjust their bidding and advertising strategies in response.

\bibliographystyle{apalike}
%\nocite{*}
\bibliography{Refer}

\newpage
\appendix

\newpage %============================================================================
\setcounter{table}{0} \setcounter{figure}{0} %
\setcounter{equation}{0} \setcounter{subsection}{0}
\setcounter{page}{1} %
\setcounter{footnote}{0} %

\renewcommand{\thetable}{A\arabic{table}} %
\renewcommand{\thepage}{Online Appendix Page \arabic{page}}
\renewcommand{\thesubsection}{A.\arabic{subsection}}%
\renewcommand{\thesubsection}{A.\arabic{subsection}}

\doublespacing
\section*{Online Appendix}
\subsection{Computation Results}\label{sec:Computational results for LP-based algorithm}
In this section, we compare the performance of our proposed LP-based algorithm and Gurobi, a commercial solver, for solving synthetic instances that mimic the SLR problem in practice. Recall that our goal is to solve the MIP formulation \hyperref[problem:MIP]{$P_0$}. Unfortunately, no MIP solver can possibly solve \hyperref[problem:MIP]{$P_0$} within the latency requirement or even close. To meet the latency requirement, we consider two approaches: (i) the proposed LP-based Algorithm \ref{al:find_opt_x}, which outputs a feasible, integral, and nearly optimal solution to \hyperref[problem:MIP]{$P_0$}; (ii) commercial LP solver Gurobi to solve the relaxed LP formulation \eqref{problem: OPT_LP}, which often outputs an infeasible and fractional solution to \hyperref[problem: OPT_LP]{$P_1$}, but it serves a good benchmark for running time comparison.
 
In all of our experiments, we simulate each instance (a given pair of $m, n, \lambda$) 1000 times and present the average results. For Gurobi, we used the concurrent method Gurobi 10.0, which runs primal simplex, dual simplex, and barrier methods simultaneously. We set the dual feasibility tolerance to be $10^{-4}$ (for Gurobi) and the $\epsilon$ to be $10^{-4}$ (for LP-based Algorithm \ref{al:find_opt_x}). The experiments run on the Dell PowerEdge R640 computing node with $20$GB of memory. It is noted that the running time for Gurobi listed below only involves the solving time and does not include the model building time. In practice, the time latency ($0.1$ seconds) is required for the sum of building and solving time. To exhibit the efficiency of the LP-based algorithm, we designed two experiments. 

% Since the LP-based algorithm has been designed to approximate the optimal solution by losing marginal approximation accuracy in the exchange of a large gain of speed, we would like to demonstrate the measurement of corresponding metrics in this experiment. Recall that the algorithm \ref{al:find_opt_mu} guarantees the feasibility of constraint \ref{eq:LP X const 1}, \ref{eq:lp X const 2}, and \ref{eq:LP X domain} but might risk the chance of violating constraint \ref{eq:LP lambda const} in some instances. Thus, in terms of measuring the approximation accuracy of the LP-based algorithm, we test the relative gap of optimal value between LP-based algorithm and Gurobi MIP, denoted by $|Opt(\text{LP}) - Opt(\text{GurobiMIP})|/Opt(\text{GurobiMIP})$, the frequency of violating the relevant constraint \ref{eq:LP X domain} and the relative violation of the relevance constraint, denoted by $|\sum_{i=1}^{m} \sum_{j=1}^{n} p_i r_j x(\text{LP})_{ij} - \lambda a|/\lambda a$. On the other hand, for algorithm efficiency, we also measure the running time improvement that LP-based algorithm gains over Gurobi.

First, due to the special structure of the formulated SLR problem, $\lambda$ plays an essential role in determining the computational cost of solving the problem. With a small $\lambda$, the relevance constraint \eqref{eq:LP lambda const} might be redundant, thus resulting in the LP-based algorithm solving extremely fast. Therefore, we would like to compare the LP-based algorithm's and Gurobi's performances with a wide range of $\lambda \in [0,1]$. In the first experiment, by fixing the problem size ($m = 50$ slots and $n = 500$ products, which is typically the largest problem size we might encounter in Marketplace A's SLR problem) and independently uniformly generating random variables ($h, r, v \sim U(0,1)$), we set up the test for multiple $\lambda$ ranging from $0$ to $1$, i.e., $\lambda \in \{.1, .2, .3, .4, .5, .6, .7, .8, .9, .925, .95, .975\}$. The optimal gap is calculated as (Optimal value of MIP \hyperref[problem:MIP]{$P_0$} - Approximated objective value of LP-based Algorithm \ref{al:find_opt_x})/Optimal value of MIP \hyperref[problem:MIP]{$P_0$}. As indicated in Table \ref{tab:Lambda}, we see that excluding the cases where $\lambda$ is too small ($\lambda \leq 0.5$), the proposed LP-based algorithm achieves more than $20$x time improvement over Gurobi LP solver. In the meantime, the LP-based algorithm can achieve the optimal gap of less than $0.05\%$ for each instance. 

%Another observation is that the LP-based algorithm especially performs better when $\lambda$ is either small (approaching $0$) or large (approaching $1$). The reason is that, since we use the bisection search algorithm for the optimal dual variable, when $\lambda$ is either small or large, the corresponding optimal dual variable is also either small or large, thus reducing the number of iterations to search for it.  

%%% first table
\vspace{0.2cm}
\begin{center}
\captionsetup{labelfont=bf}
\captionof{table}{Experiment Results for Multiple $\lambda$, $m = 50$, and $n = 500$} \label{tab:Lambda} 
\scalebox{.775}{
    \begin{tabular}{c|cccccccccccc}
    \hline
    $\lambda$ &.1 &.2 &.3 &.4 &.5 &.6 &.7 &.8 &.9 &.925 &.95 &.975 \\
    \hline
    %Optimal value gap (\%) &.00 &.00 &.00 &.00 &.00 &.00 &.01 &.02 &.03 &.04 &.05 &.06 &.27 \\
    Optimality Gap (\%) 
   & 0.000* & 0.000* & 0.000* & 0.000* & 0.000* & 0.001 & 0.003 & 0.008 & 0.015 & 0.019 & 0.027 & 0.042 \\
    %& .017 & .017 & .017 & .017 & .017 & .019 & .026 & .037 & .063 & .078 & .100 & .137 \\ 
    LP-based time (seconds) 
    & 0.001 & 0.001 & 0.001 & 0.001 & 0.004 & 0.009 & 0.01 & 0.01 & 0.01 & 0.01 & 0.011 & 0.012 \\
    %& 0.001 & 0.001 & 0.001 & 0.001 & 0.004 & 0.01 & 0.01 & 0.01 & 0.01 & 0.011 & 0.011 & 0.013 \\
    %& .001 & .001 & .001 & .002 & .004 & .010 & .010 & .010 & .010 & .011 & .011 & .012 \\
    %&* &* &* &* &.015 &.037 &.037 &.034 &.026 &.024 &.020 &.015 \\
    Gurobi time (seconds) 
    & 0.224 & 0.205 & 0.204 & 0.205 & 0.212 & 0.231 & 0.247 & 0.264 & 0.3 & 0.317 & 0.337 & 0.362 \\
    %& 0.366 & 0.359 & 0.359 & 0.357 & 0.391 & 0.49 & 0.554 & 0.579 & 0.62 & 0.635 & 0.659 & 0.689 \\
    %& .211 & .21 & .210 & .210 & .215 & .233 & .249 & .267 & .298 & .310 & .327 & .343  \\
    %&.327 &.324 &.321 &.319 &.320 &.316 &.321 &.319 &.321 &.321 &.324 &.324 \\
    \begin{tabular}{@{}c@{}}Time improvement \\ \scalebox{.8}{(Gurobi time/LP-based time)}\end{tabular} 
    & 289** & 269** & 262** & 252** & 48** & 24 & 26 & 27 & 30 & 31 & 30 & 30  \\
    %& 429* & 444* & 435* & 408* & 93* & 50 & 55 & 58 & 61 & 60 & 58 & 54 
    %& 143* & 154* & 144* & 138* & 50 & 23 & 25 & 27 & 29 & 29 & 29 & 28 
    %& **  & **  & **  & **  & 21.230  & 8.615  & 8.717  & 9.491  & 12.130  & 13.410  & 16.280  & 21.029   
    \\
    \bottomrule 
    \end{tabular}
}
% \begin{tablenotes}
%       \small
%       \item \textit{Notes.} *  When $\lambda$ is small, then the relevance constraint \eqref{eq:LP lambda const} is redundant, thus the LP-based algorithm will generate the optimal solutions.
%       \item ** The time improvement is not representative when the relevance constraint \eqref{eq:LP lambda const} is redundant, as Gurobi might not be able to distinguish this case.
% \end{tablenotes}
\end{center}
\vspace{0.2cm}

In the second experiment, we would like to explore the performance of the LP-based algorithm for various problem sizes, especially when the problem dimension increases. Thus, for the second experiment, by picking $\lambda = 0.95$ (which is a reasonable value for $\lambda$ as shown in the field study), we test with a range of $m \in \{10, 20, 50, 100, 200, 500\}$ and $n \in \{50, 100, 200, 500, 1000, 2000\}$ ($m \leq n$).
As shown in Table \ref{tab:Experiment results for various problem sizes}, compared to Gurobi, the LP-based algorithm achieves a massive improvement, especially under large-scale problems. We see that the time improvement grows as the problem size increases, whereas the optimal gap reduces. When $m = 500$, the LP-based algorithm is more than $100$x faster than Gurobi with the optimal gap to be $0.004\%$. Under the $0.1$ seconds latency requirement, we see that the LP-based solver can solve a problem with a size of at most $m = 500, n = 500$. Yet, under the same conditions, Gurobi can only solve one with a size of $m = 100, n = 100$. Overall, the optimal gap is usually negligible ($< 1\%$) for reasonably sized instances.

The above two experiments demonstrate the efficiency of the proposed algorithm compared to state-of-the-art commercial solvers. Indeed, the LP commercial solver is not a feasible solution due to the latency requirement of the marketplace, while our proposed Algorithm \ref{al:find_opt_x} enables such an LP approach for sponsored listing.

%%% second table
\vspace{0.2cm}
\begin{center}
\begin{table}[h]
\captionsetup{labelfont=bf}
\captionof{table}{Experiment Results for Various Problem Sizes} \label{tab:Experiment results for various problem sizes} 
%\begin{threeparttable}[]
\footnotesize %\setlength{\tabcolsep}{2.5pt}
%\resizebox{\linewidth}{!}{
\scalebox{.925}{
\begin{tabular}{c|cccccc} 
\toprule
\multicolumn{7}{c}{Panel A: LP-based time (seconds)} \\ \midrule
% \diagbox{m}{n} &50 &100 &200 &500 &1000 &2000 \\ \hline
%     10  & 0.002 & 0.004 & 0.008 & 0.021 & 0.044 & 0.095 \\
%     20  & 0.002 & 0.004 & 0.007 & 0.018 & 0.038 & 0.081 \\
%     50  & 0.004 & 0.006 & 0.009 & 0.018 & 0.037 & 0.078 \\
%     100 &       & 0.006 & 0.010 & 0.022 & 0.049 & 0.110 \\
%     200 &       &       & 0.013 & 0.036 & 0.083 & 0.179 \\
%     500 &       &       &       & 0.101 & 0.215 & 0.438 \\ \bottomrule
% \end{tabular}
\diagbox{m}{n} 
    & 50 & 100 & 200 & 500 & 1000 & 2000 \\ \hline

% 10 & .001 & .001 & .003 & .007 & .012 & .024 \\ 

% 20 & .001 & .002 & .004 & .009 & .015 & .029 \\ 

% 50 & .001 & .003 & .005 & .012 & .023 & .046 \\ 

% 100 & & .001 & .007 & .018 & .037 & .075 \\ 

% 200 & & & .002 & .030 & .063 & .136 \\ 

% 500 & & & & .008 & .134 & .282 \\ 

10 & 0.001 & 0.001 & 0.003 & 0.006 & 0.011 & 0.022 \\ 

20 & 0.001 & 0.002 & 0.003 & 0.007 & 0.013 & 0.027 \\ 

50 & 0.001 & 0.002 & 0.005 & 0.011 & 0.021 & 0.041 \\ 

100 & & 0.003 & 0.006 & 0.017 & 0.034 & 0.067 \\ 

200 & & & 0.011 & 0.027 & 0.061 & 0.128 \\ 

500 & & & & 0.061 & 0.125 & 0.283 \\ 

\bottomrule
\end{tabular}
\hfill
\begin{tabular}{c|cccccc} 
\toprule
\multicolumn{7}{c}{Panel B: Gurobi time (seconds)} \\ \midrule
    \diagbox{m}{n} 
    % &50 &100 &200 &500 &1000 &2000 \\ \hline
    % 10   & 0.003  & 0.006  & 0.012  & 0.034  & 0.092  & 0.268  \\
    % 20   & 0.007  & 0.013  & 0.024  & 0.082  & 0.206  & 0.527  \\
    % 50   & 0.018  & 0.043  & 0.103  & 0.375  & 0.729  & 1.756  \\
    % 100  &        & 0.091  & 0.278  & 0.561  & 1.017  & 2.079  \\
    % 200  &        &        & 1.023  & 1.969  & 3.451  & 6.606  \\
    % 500  &        &        &        & 13.453  & 27.908  & 42.156  \\
    % \bottomrule
    & 50 & 100 & 200 & 500 & 1000 & 2000 \\ \hline

% 10 & .035 & .040 & .061 & .112 & .209 & .484 \\ 

% 20 & .052 & .062 & .104 & .233 & .485 & 1.200 \\ 

% 50 & .027 & .097 & .215 & .642 & 1.554 & 3.903 \\ 

% 100 & & .129 & .412 & .974 & 2.121 & 4.654 \\ 

% 200 & & & .709 & 2.630 & 5.264 & 11.374 \\ 

% 500 & & & & 7.245 & 23.402 & 44.520 \\ 

% 10 & 0.033 & 0.038 & 0.056 & 0.112 & 0.21 & 0.484 \\ 

% 20 & 0.053 & 0.069 & 0.107 & 0.229 & 0.48 & 1.225 \\ 

% 50 & 0.052 & 0.103 & 0.227 & 0.673 & 1.65 & 4.282 \\ 

% 100 & & 0.197 & 0.415 & 1.025 & 2.229 & 4.896 \\ 

% 200 & & & 1.153 & 2.757 & 5.527 & 11.631 \\ 

% 500 & & & & 12.777 & 23.404 & 43.632 \\ 

10 & 0.003 & 0.006 & 0.012 & 0.038 & 0.102 & 0.326 \\ 

20 & 0.007 & 0.014 & 0.027 & 0.085 & 0.225 & 0.671 \\ 

50 & 0.017 & 0.045 & 0.110 & 0.343 & 0.851 & 2.304 \\ 

100 & & 0.094 & 0.205 & 0.423 & 0.844 & 1.832 \\ 

200 & & & 0.639 & 1.366 & 2.468 & 5.170 \\ 

500 & & & & 8.158 & 16.339 & 28.717 \\ 

    \bottomrule
\end{tabular}
}
\bigskip

% \begin{tabular}{c|cccccc} 
% \toprule
% \multicolumn{7}{c}{Panel C: Gurobi time (seconds)} \\ \midrule
%     \diagbox{m}{n} &50 &100 &200 &500 &1000 &2000 \\ \hline
%     10& 0.18  & 0.23  & 0.22  & 0.51  & 0.85  & 1.89  \\
%     20& 0.18  & 0.32  & 0.41  & 1.36  & 1.89  & 3.79  \\
%     50& 0.28  & 0.36  & 0.91  & 2.27  & 4.67  & 9.72  \\
%     100& & 0.88  & 1.93  & 4.99  & 7.95  & 13.73  \\
%     200& & & 2.77  & 7.24  & 14.19  & 29.83  \\
%     500& & & & 20.56  & 48.76  & 104.48  \\ \bottomrule
% \end{tabular}
% \hfill
\scalebox{.97}{
%\begin{center}
\begin{tabular}{c|cccccc} 
\toprule
\multicolumn{7}{c}{Panel C: Time improvement (Gurobi time/LP time)} \\ \midrule
    \diagbox{m}{n} 
    % &50 &100 &200 &500 &1000 &2000 \\ \hline
    % 10   & 1.286  & 1.533  & 1.587   & 1.626   & 2.073   & 2.807  \\
    % 20   & 3.158  & 3.401  & 3.251   & 4.658   & 5.456   & 6.535  \\
    % 50   & 5.037  & 7.811  & 11.460  & 20.435  & 19.938  & 22.609 \\
    % 100  &        & 14.152 & 27.468  & 25.224  & 20.707  & 18.837 \\
    % 200  &        &        & 76.767  & 54.936  & 41.440  & 36.948 \\
    % 500  &        &        &         & 133.519 & 129.674 & 96.144 \\ 
   & 50 & 100 & 200 & 500 & 1000 & 2000 \\ \hline

% 10 & 52 & 27 & 18 & 15 & 17 & 20 \\ 

% 20 & 69 & 35 & 26 & 27 & 33 & 41 \\ 

% 50 & 81 & 38 & 39 & 52 & 67 & 85 \\ 

% 100 & & 164 & 56 & 55 & 57 & 62 \\ 

% 200 & & & 422 & 89 & 83 & 84 \\ 

% 500 & & & & 900 & 175 & 158 \\ 

% 10 & 51.4 & 29.4 & 21.4 & 19.0 & 19.0 & 21.7 \\ 

% 20 & 72.2 & 44.7 & 33.6 & 32.2 & 36.0 & 46.1 \\ 

% 50 & 48.4 & 46.2 & 48.5 & 61.4 & 79.1 & 104.5 \\ 

% 100 & & 60.8 & 64.1 & 59.9 & 64.9 & 72.9 \\ 

% 200 & & & 105.1 & 100.9 & 90.4 & 90.9 \\ 

% 500 & & & & 209.4 & 186.6 & 154.2 \\ 

10 & 4.9 & 4.7 & 4.8 & 6.4 & 9.2 & 14.6 \\ 

20 & 9.2 & 9.1 & 8.6 & 11.9 & 16.9 & 25.3 \\ 

50 & 16.0 & 20.1 & 23.6 & 31.4 & 40.8 & 56.2 \\ 

100 & & 29.2 & 31.7 & 24.7 & 24.6 & 27.3 \\ 

200 & & & 58.3 & 50.0 & 40.4 & 40.4 \\ 

500 & & & & 133.7 & 130.3 & 101.5 \\ 

    \bottomrule
\end{tabular}
\hfill
\begin{tabular}{c|cccccc} 
\toprule
\multicolumn{7}{c}{Panel D: Optimal Gap (\%)} \\ \midrule
    \diagbox{m}{n} 
    % &50 &100 &200 &500 &1000 &2000 \\ \hline
    % 10   & 1.286  & 1.533  & 1.587   & 1.626   & 2.073   & 2.807  \\
    % 20   & 3.158  & 3.401  & 3.251   & 4.658   & 5.456   & 6.535  \\
    % 50   & 5.037  & 7.811  & 11.460  & 20.435  & 19.938  & 22.609 \\
    % 100  &        & 14.152 & 27.468  & 25.224  & 20.707  & 18.837 \\
    % 200  &        &        & 76.767  & 54.936  & 41.440  & 36.948 \\
    % 500  &        &        &         & 133.519 & 129.674 & 96.144 \\ 
% & 50 & 100 & 200 & 500 & 1000 & 2000 \\ \hline

% 10 & 1.094 & 0.730 & 0.410 & 0.193 & 0.105 & 0.049 \\ 

% 20 & 0.559 & 0.360 & 0.253 & 0.107 & 0.053 & 0.026 \\ 

% 50 & 0.261 & 0.295 & 0.204 & 0.093 & 0.045 & 0.016 \\ 

% 100 & & 0.194 & 0.252 & 0.127 & 0.075 & 0.027 \\ 

% 200 & & & 0.193 & 0.212 & 0.135 & 0.061 \\ 

% 500 & & & & 0.189 & 0.230 & 0.160 \\ 
& 50 & 100 & 200 & 500 & 1000 & 2000 \\ \hline

% 10 & 1.005 & 0.682 & 0.405 & 0.187 & 0.094 & 0.049 \\ 

% 20 & 0.355 & 0.283 & 0.192 & 0.087 & 0.049 & 0.025 \\ 

% 50 & 0.067 & 0.064 & 0.056 & 0.031 & 0.019 & 0.01 \\ 

% 100 & & 0.018 & 0.018 & 0.012 & 0.008 & 0.005 \\ 

% 200 & & & 0.005 & 0.004 & 0.003 & 0.002 \\ 

% 500 & & & & 0.001 & 0.001 & 0.001 \\ 

10 & 0.83 & 0.573 & 0.343 & 0.159 & 0.078 & 0.039 \\ 

20 & 0.336 & 0.269 & 0.182 & 0.081 & 0.044 & 0.020 \\ 

50 & 0.062 & 0.059 & 0.052 & 0.027 & 0.015 & 0.007 \\ 

100 & & 0.015 & 0.015 & 0.01 & 0.006 & 0.004 \\ 

200 & & & 0.004 & 0.004 & 0.004 & 0.004 \\ 

500 & & & & 0.004 & 0.004 & 0.004 \\
    \bottomrule
\end{tabular}
%\end{center}
}

\end{table}
\end{center}
\vspace{0.2cm}

\subsection{Guidance for Fine-Tuning the LP-Based Algorithm}\label{subsec:guidience}
This section provides practical guidelines for fine-tuning and implementing the LP-based algorithm. The algorithm involves only one parameter, the relevance parameter $\lambda$. As a first step, we recommend that readers estimate a suitable range of $\lambda$ for subsequent A/B testing using historical data.  

For each historical impression $k$ generated by the current ranking algorithm, assume access to $\boldsymbol{v}$, $\boldsymbol{r}$, $\boldsymbol{h}$, and the corresponding ranking $\boldsymbol{X}^{k}$. These inputs allow the computation of a relevance score (see Equation \ref{eq:MIP lambda const}):  
\begin{align}
    \text{Rel}(\boldsymbol{X}^{k})
    =
    \frac{\sum_{i=1}^{m} \sum_{j=1}^{n} h_i r_j \boldsymbol{X}_{ij}^{k}}{a}.
\end{align}

The intuition is as follows. Substituting $\text{Rel}(\boldsymbol{X}^{k})$ into Equation \ref{eq:MIP lambda const} yields $a(\text{Rel}(\boldsymbol{X}^{k}) - \lambda) \geq 0$, which simplifies to $\text{Rel}(\boldsymbol{X}^{k}) \geq \lambda$. In other words, $\lambda$ specifies a lower bound on the relevance score for future impressions.  

We suggest that practitioners construct the empirical distribution of $\text{Rel}(\boldsymbol{X}^{k})$ across historical impressions and select a summary statistic (e.g., the mean, median, or a chosen quantile) as a benchmark value, denoted $\lambda_{0}$. Implementing the LP-based algorithm with $\lambda_{0}$ ensures that the resulting relevance scores are, on average, comparable to those under the prior ranking algorithm. This provides a principled starting point for further tuning via A/B testing.

Building on this benchmark, readers can adjust $\lambda$ according to their strategic objectives. For those aiming to improve ranking quality relative to the current strategy, we recommend starting with $\lambda_{0}$ and testing incremental increases. Specifically, testing groups can be formed with relevance parameters $\lambda_{0}, \lambda_{0} + \delta, \lambda_{0} + 2\delta, \ldots$ for some user-chosen step size $\delta > 0$. As shown in our earlier analysis, higher values of $\lambda$ typically yield rankings that are more relevant.  

Conversely, practitioners seeking to prioritize revenue may explore decreasing $\lambda$ from $\lambda_{0}$, for example, $\lambda_{0}, \lambda_{0} - \delta, \lambda_{0} - 2\delta, \ldots$. Ultimately, the optimal choice of $\lambda$ should be guided by the outcomes of A/B testing, balancing the trade-off between relevance and revenue in line with organizational objectives.

\subsection{Proof of Theorem~\ref{thm:optimality_lp_algo}}\label{subsec:theorem}
In order to prove Theorem~\ref{thm:optimality_lp_algo}, we first show the following lemma that gives the closed-form solution of the LP-relaxed primal problem \hyperref[problem: OPT_LP]{$P_1$}:
\begin{lem}[Closed-form solution of the LP-relaxed primal problem] Under \Cref{ass:decreasing position weights,ass:unique_rank_order}
the optimal solution of the LP-relaxed primal problem \hyperref[problem: OPT_LP]{$P_1$}, denoted by $\boldsymbol{X}^{LP}$, has the following closed form:
    \begin{enumerate}
        \item \textbf{Case A (Nonbinding relevance constraint).} Fix $\mu = 0$ with $\boldsymbol{\delta} = \boldsymbol{v}$, then if \\
        $\text{Rel}(\text{PrimalResponse}(\boldsymbol{\delta},m)) \geq \lambda a$, then the relevance constraint is nonbinding with $\mu^* = 0$, and the LP optimal solution is integral: $\boldsymbol{X}^{LP} = \text{PrimalResponse}(\boldsymbol{\delta},m)$.
        
        \item \textbf{Case B (Binding relevance constraint).} The relevance constraint binds when $\mu^* > 0$. For any dual variable $\mu$, we define its associated integral solution as \\
        $\boldsymbol{X}_{+}(\mu) := \lim_{t \rightarrow 0}\text{PrimalResponse}(\boldsymbol{v} + (\mu+t)\boldsymbol{r},m)$, then the optimal dual variable $\mu^*$ of the problem \hyperref[problem: OPT_LP_Dual]{$D_1$} is characterized as the smallest pairwise indifference value such that its associated integral solution is feasible, i.e., : 
        \begin{equation}
            \mu^* = \inf \{\mu \in \{0\} \cup \{\mu_{jk}\}_{(j,k) \in \mathcal{I}} : \text{Rel}(\boldsymbol{X}_{+}(\mu)) \geq \lambda a\}.
        \end{equation}
        Then for such $\mu^*$, let $\boldsymbol{\delta}^* := \boldsymbol{v} + \mu^* \boldsymbol{r}$ and its sorted indices in non-increasing order as $j_{(1)},\cdots,j_{(n)}$.
        \begin{itemize}
            \item \textbf{B1 (Boundary tie between position $m$ and $m+1$).} If $\delta^*_{j(m)} = \delta^*_{j(m+1)}$, then the optimal solution is fractional only in the last row:
            \begin{equation}
                \boldsymbol{X}^{\text{LP}}_{i,j(i)} = 1 \; \forall \; i < m,
                \boldsymbol{X}^{\text{LP}}_{m,j(m)} = \alpha,
                \boldsymbol{X}^{\text{LP}}_{m,j(m+1)} = 1 - \alpha, 
                \boldsymbol{X}^{\text{LP}}_{ij} = 0 \text{ otherwise},
            \end{equation}
            %where $\beta = \frac{\lambda a - \sum_{i = 1}^{m - 1} h_{i} r_{j(i)} - h_{m} r_{j(m)}}{h_{m}(r_{j(m+1)} - r_{j(m)})}$.

            \item \textbf{B2 (Interior tie within the top $m$ positions).} If $\delta^*_{j(s)} = \delta^*_{j(s+1)}$ for some $s < m$, then optimal solution is fractional only in row $s$ and $s+1$:
            \begin{align}
                &\boldsymbol{X}^{\text{LP}}_{i,j(i), } = 1 \; \forall \; i \in [m]\setminus\{s,s+1\},
                \boldsymbol{X}^{\text{LP}}_{s,j(s)} = \boldsymbol{X}^{\text{LP}}_{s+1,j(s+1)} = 1 - \alpha,
                \nonumber
                \\
                &\boldsymbol{X}^{\text{LP}}_{s,j(s+1)} = \boldsymbol{X}^{\text{LP}}_{s+1,j(s)} = \alpha,
                \boldsymbol{X}^{\text{LP}}_{ij} = 0 \text{ otherwise},
            \end{align}
            %where $\beta = \frac{\lambda a - \sum_{i = 1, i \notin \{s, s+1\}}^{m} h_i r_{j(i)} - h_{s} r_{j(s+1)} - h_{s+1} r_{j(s)}}{(h_{s} - h_{s+1})(r_{j(s)} - r_{j(s+1)})}$.

            \item \textbf{B3 (No tie within the top $m$ positions).} If $\delta^*_{j(1)} > \cdots > \delta^*_{j(m)}$, i.e., all top-m scores are strictly ordered, then the optimal LP solution is integral:
            \begin{equation}
                \boldsymbol{X}^{\text{LP}}_{i,j(i)} = 1 \; \forall \; i \in [m],
                \boldsymbol{X}^{\text{LP}}_{ij} = 0 \text{ otherwise}.
            \end{equation}
        \end{itemize}
    \end{enumerate}
\end{lem}

\begin{proof}
    For \textbf{Case A}, it is trivial to see that when the relevance constraint of \hyperref[problem: OPT_LP]{$P_1$} is nonbinding, then the optimal dual variable $\mu^* = 0$. The LP optimal solution is integral and the permutation of the sorted vector $\boldsymbol{v}$ in a 
    descending order. 

    For \textbf{Case B}, we show that the derived LP solutions in all sub-cases are (1) feasible and (2) optimal. 

    \textbf{Feasibility Check.} 
    \begin{itemize}
        \item \textbf{B1 (Boundary tie between position $m$ and $m+1$).}
        \begin{align}
            \text{Rel}(\boldsymbol{X}^{\text{LP}})
            & =
            \sum_{i = 1}^{m - 1} h_{i} r_{j(i)} + \alpha h_{m} r_{j(m)} + (1 - \alpha) h_{m} r_{j(m+1)}
            \\
            &=
            \sum_{i = 1}^{m - 1} h_{i} r_{j(i)} + \lambda a - \sum_{i = 1}^{m - 1} h_{i} r_{j(i)}
            \\
            & = 
            \lambda a.
        \end{align}
        
        \item \textbf{B2 (Interior tie within the top $m$ positions).}
        \begin{align}
            \text{Rel}(\boldsymbol{X}^{\text{LP}})
            & =
            \sum_{i = 1, i \notin \{s, s+1\}}^{m} h_i r_{j(i)} + (1 - \alpha) (h_s r_{j(s)} +h_{s+1} r_{j(s+1)}) + \alpha (h_{s+1} r_{j(s)} +h_{s} r_{j(s+1)})
            \\
            &=
            \sum_{i = 1, i \notin \{s, s+1\}}^{m} h_i r_{j(i)} + \lambda a - \sum_{i = 1, i \notin \{s, s+1\}}^{m} h_i r_{j(i)}
            \\
            & = 
            \lambda a.
        \end{align}        

        \item \textbf{B3 (No tie within the top $m$ positions).} The proposed $\boldsymbol{X}^{\text{LP}}$ is the permutation matrix that assigns the top $m$ items to the $m$ rows in the order of $\boldsymbol{h}$. We will see below that in this subcase necessarily $\text{Rel}(\boldsymbol{X}^{\text{LP}}) = \lambda a$, otherwise $\mu^*$ would not be dual-optimal.
    \end{itemize}

    \textbf{Optimality Check.} 
    
    We show that $\boldsymbol{X}^{\text{LP}}$ in all three subcases of \textbf{Case B} are optimal for \hyperref[problem: OPT_LP]{$P_1$}. We first verify the KKT conditions for the relaxed LP problem: 
    \begin{enumerate}
        \item \textbf{Dual feasibility:} by definition of the dual problem.
        
        \item \textbf{Primal feasibility:} as we have shown in the previous part of \textbf{Feasibility check}.
        
        \item \textbf{Stationarity:} for all three subcases, we have $\boldsymbol{X}^{\text{LP}} \in \argmax_{\boldsymbol{X} \in \mathcal{F}} L(\mu, \boldsymbol{X})$. In case B1 and B2, every split within the tied block maximizes $L(\mu^*, \cdot)$, so our $\boldsymbol{X}^{\text{LP}}$ is one such maximizer; in case B3, since the order of $\boldsymbol{\delta}^*$ is strict over $j(1),\cdots,j(m)$, the stated permutation is the unique maximizer.

        \item \textbf{Complementary slackness:} for all three cases, we have $\mu^*(\lambda a - \text{Rel}(\boldsymbol{X}^{\text{LP}})) = 0$.
    \end{enumerate}
    Since all KKT conditions are satisfied, therefore, $\boldsymbol{X}^{\text{LP}}$ is optimal for \hyperref[problem: OPT_LP]{$P_1$}. 
\end{proof}

Then we have the following proof for \Cref{thm:optimality_lp_algo}:
\begin{proof}
    Note that we could re-write $\mathbb{E}[\Tilde{\boldsymbol{X}}] = (1 - \alpha) \boldsymbol{X}_{+} + \alpha \boldsymbol{X}_{-}$. Then we discuss over each case.

    \textbf{Case A (Nonbinding relevance constraint).} With $\mu^* = 0$, we have $\boldsymbol{X}_{-} = \boldsymbol{X}_{+} = \boldsymbol{X}^{\text{LP}}$. Thus $\mathbb{E}[\Tilde{\boldsymbol{X}}] = \boldsymbol{X}^{\text{LP}}$, where $\boldsymbol{X}^{\text{LP}}$ is the optimal solution to \hyperref[problem: OPT_LP]{$P_1$}.

    \textbf{Case B1 (Boundary tie between position $m$ and $m+1$).} From \Cref{al:find_opt_mu}, we see that both $\boldsymbol{X}_{-}$ and $\boldsymbol{X}_{+}$ have the exact same first rows $m-1$ with $\boldsymbol{X}_{+, i,j(i)} = \boldsymbol{X}_{-, i,j(i)} = 1$ for all $i < m$. Thus, $\mathbb{E}[\Tilde{\boldsymbol{X}}_{i,j(i)}] = \boldsymbol{X}^{\text{LP}}_{i,j(i)} = 1$ for all $i < m$. In the $m$-th row, $\boldsymbol{X}_{+, m,j(m+1)} = 1$ while $\boldsymbol{X}_{-, m,j(m)} = 1$. 

    Thus, we get the $m$-th row of $\mathbb{E}[\Tilde{\boldsymbol{X}}]$ as:
    \begin{equation}
        \mathbb{E}[\Tilde{\boldsymbol{X}}_{m,j(m)}] 
        = \alpha 
        = \boldsymbol{X}^{\text{LP}}_{m,j(m)},
    \end{equation}
    and 
    \begin{equation}
        \mathbb{E}[\Tilde{\boldsymbol{X}}_{mj(m+1)}] 
        = 1 - \alpha 
        = \boldsymbol{X}^{\text{LP}}_{m,j(m+1)}.
    \end{equation}
    Thus, $\mathbb{E}[\Tilde{\boldsymbol{X}}] = \boldsymbol{X}^{\text{LP}}$.

    \textbf{B2 (Interior tie within the top $m$ positions).} Similarly, we see that both $\boldsymbol{X}_{-}$ and $\boldsymbol{X}_{+}$ have the exact same rows except for $s$- and $s+1$-th row. Specifically, $\boldsymbol{X}_{-, s,j(s)} = 1$ and $\boldsymbol{X}_{+, s+1,j(s+1)} = 1$. 

    Therefore, we have $\mathbb{E}[\Tilde{\boldsymbol{X}}_{s,j(s)}] = \mathbb{E}[\Tilde{\boldsymbol{X}}_{s+1,j(s+1)}] = 1 - \alpha$ and $\mathbb{E}[\Tilde{\boldsymbol{X}}_{s,j(s+1)}] = \mathbb{E}[\Tilde{\boldsymbol{X}}_{s+1,j(s)}] = \alpha$. This gives: $\mathbb{E}[\Tilde{\boldsymbol{X}}] = \boldsymbol{X}^{\text{LP}}$.

    \textbf{B3 (No tie within the top $m$ positions).} Then the Lagrangian maximizer at $\mu^*$ is unique with $\boldsymbol{X}_{+} = \boldsymbol{X}_{-} = \boldsymbol{X}^{\text{LP}}$. Thus $\mathbb{E}[\Tilde{\boldsymbol{X}}] = \boldsymbol{X}^{\text{LP}}$.
\end{proof}

\subsection{Summary Statistics of the Sampled Data}\label{subsec:sampledata}

Panel A in table \ref{tab:summarystat_online} reports the summary statistics for our second dataset, introduced in Section \ref{subsec:data}, which consists of a random sample from the full-scale experiment. Panel B in table \ref{tab:summarystat_online} presents the summary statistics for our third dataset, a random sample drawn prior to the launch of the field experiment. Both samples are at the impression level and we report summary statistics only for those variables whose disclosure is permitted by the company. All monetary variables are reported in USD.

\vspace{.2in}
\begin{table}[ht]
\centering
\caption{Summary Statistics.}
\begin{tabular}{ l c c c c}
 \hline\hline
 Variable  & Mean & Median & SD & Sample Size\\
 \hline
 \multicolumn{5}{c}{Panel A: Random Sample of the Full-Scale Experiment} \\
 \hline
 Average Product Price & 2.69 & 0.31 & 57.33 &2,071,833 \\ 
 %Weighted PTR & \\
 %Total Number of Items & 10.93 & 11.00 & 1.21 & 2,071,833\\
 Number of Clicks & 1.07 & 1.00 &  0.53 & 2,071,833\\
 %Number of Purchases & 0.19 & 0.00 & 0.14 & 2,071,833\\
 %Revenue & 0.02 & 0.00 & 0.14 & 2,071,833\\
 %Transaction Price & 1.55 & 1.19 & 1.22 & 390,347 \\
 %Ranking of Transaction Product & 3.08 & 2.00 & 2.43 & 390,347\\
 \hline
 \multicolumn{5}{c}{Panel B: Random Sample Drawn Before the Experiment} \\
 \hline
 Average Product Price & 301.31 & 29.69 & 11223.03 & 13,370,159\\
 Total Number of Candidate Items & 70.17 & 60.00 & 47.90 & 266,003 \\
 Total Number of Items & 11.34 & 12.00 & 2.35 & 266,003\\
 Total Number of Sponsored Items & 9.43 & 12.00 & 4.11 & 266,003\\
 Number of Clicks & 0.32 & 0.00 & 0.64 & 266,003\\
 %Number of Purchases & 0.01 & 0.00 & 0.10 & 341,399\\
 %Revenue & 0.08 & 0.00 & 3.84 & 341,399\\
 %Transaction Price & 37.19 & 14.80 & 148.07 & 2,698\\
 %Ranking of Transaction Product & 3.95 & 2.00 & 5.97 & 2,698\\
 Seed Item Price & 819.38 & 28.99 & 53,169.45 & 266,003\\
 %Seed Item Revenue & 0.62 & 0.00 & 7.18 & 84,919\\
 \hline\hline
\end{tabular}
\label{tab:summarystat_online}
\end{table}

\subsection{Extension: Overall Planning}\label{sec:overall planning}
In practice, there may be additional requirements from marketplaces and/or customers. For instance, (1) some marketplaces would like to pose constraints such as inventory limits or targeted revenue for each seller or product category, and (2) some consumers way prefer the diversity of items shown to them. These global constraints involve many impressions over a period; thus, we call it overall planning. Unfortunately, there is no guarantee that a score-based algorithm can satisfy such constraints because a score-based algorithm takes greedy operations for local objectives. On the other hand, these constraints can be naturally incorporated into an LP-based algorithm, as described below.

% However, directly imposing these extra constraints into \hyperref[problem:MIP]{$P_0$}  would increase the complexity of solving, potentially breaching the latency requirement. Therefore, in this section, we discuss how the idea of overall planning can be incorporated into this LP framework to help address these challenges.
\subsubsection{LP-Based Algorithm for Overall Planning}

The basic idea is to utilize historical data to obtain corresponding dual variables (i.e., reduced costs) associated with the overall planning constraints and then utilize these dual variables in the online LP-based ranking.

\textbf{Step 1: Obtain dual variables for global constraints.}
Given a historical dataset that records the impressions for each day, suppose there are $T$ impressions, $G$ unique consumers, and a group of $K$ selected sellers who pay additional \textit{Ad fees} and accordingly pose requirements on inventory and revenue. For each impression $t$, it has $m_t$ slots and $n_t$ candidate items. Let $h_{it}, v_{jt}, r_{jt}$ denote the associated position weight for each slot $it$, revenue and relevance for each candidate item $jt$. Similar to the previous setting, for each impression $t$, we can get the maximum relevance score $a_t$ that any ranking can obtain, and define any relevance hyperparameter $\lambda_t \in [0,1]$.

For a selected seller $k$, there is an inventory limit $L_k$, and a target revenue $I_k$ within a time period (usually a day). We consider the following overall planning constraints: (i) from the perspective of market fairness, the online marketplace does not interfere with the number of items a single seller $k$ could display to a single consumer $g$; (ii) for each consumer $g$, she would prefer candidate items from various diverse sellers instead of a few who have paid high \textit{Ad fees}. Therefore, we would set up a lower bound of the total of sellers whose items are displayed to her, denoted by $S_g$. (iii) After paying an extra Ad rate, the seller $k$ would want to display his items to at least $C_k$ number of consumers. We would highlight that any combination of these constraints (even some similar constraints) can be incorporated into this algorithmic framework.

Suppose we have $T$ impressions in the historical data, indexed by $t$,i.e., $t = 1,...,T$. We select $K$ unique sellers, indexed by $k = 1,...,K$, and $G$ unique consumers, indexed by $g = 1,...,G$. Each consumer $g$ has total $T_g$ impressions. Moreover, among these $T_g$ impressions, the candidate items are provided by $K_g$ sellers. 
%For consumer $g$, there exists a set $\mathcal{Q}_g$ whose $l$th element ${q_g}_k$ denotes the total number of displayed items sold by seller $k$. The set containing all indices of $k$'s candidate items is denoted by $\mathcal{D}_{k}$. Naturally, ${q_g}_k = \sum_{t}^{T_g} x_{ijt} \mathds{1}\{j \in \mathcal{D}_{k}\}$. We can see that if no items from seller $k$ have been displayed to any consumers, then ${q_g}_k = 0$. Thus the total number of sellers that each consumer $g$ gets exposed to is $\sum_{k}^{K_g}\mathds{1}\{{q_g}_k > 0\}$. Similarly, for each seller $k$, we say there exists a set $\mathcal{O}_k$ whose $l$th element ${o_k}_l$ denotes the total number of items displayed to consumer $g$. Naturally, ${o_k}_g = \sum_{t}^{T_g} x_{ijt} \mathds{1}\{j \in \mathcal{D}_{l}\}$. Thus the total number of consumers that each seller $k$ gets exposed to is $\sum_{g}^{G}\mathds{1}\{{o_k}_g > 0\}$.

Using historical data, we can obtain the optimal ranking with overall constraints (i), (ii), (iii) by solving the following LP:
\begingroup
\allowdisplaybreaks
\begin{align}
    \max_{\boldsymbol{X}} & \sum_{t=1}^T \sum_{i=1}^{m_t} \sum_{j=1}^{n_t} h_{it} v_{jt} \boldsymbol{X}_{ijt}\\
    \text{s.t. } 
    & 
    \sum_{i=1}^{m_t} \sum_{j=1}^{n_t} h_{it} r_{jt} \boldsymbol{X}_{ijt} \geq \lambda_{t} a_{t}, \forall t=1,...,T \label{eq:lambda const} 
    \\
    & 
    \sum_{t=1}^{T} \sum_{i=1}^{m_t} \sum_{j=1}^{n_t} \boldsymbol{X}_{ijt} \mathbbm 1\{\text{item } jt \text{ listed by seller } k\} \leq L_k, \forall k=1,...,K \label{eq:inventory}
    \\
    & 
    \sum_{t=1}^{T} \sum_{i=1}^{m_t} \sum_{j=1}^{n_t} h_{it} \boldsymbol{X}_{ijt} v_{jtk} \mathbbm 1\{\text{item } jt \text{ listed by seller } k\}  \geq I_k, \forall k=1,...,K\label{eq:target}
    \\
    & \sum_{g = 1}^{G} \sum_{t=1}^{T_g} \sum_{i=1}^{m_t} \sum_{j=1}^{n_t} \boldsymbol{X}_{ijt}  \mathbbm 1\{\text{item } jt \text{ listed by seller } k\} \geq C_k, \forall k = 1,\cdots,K
    \label{eq:minimum seller}
    \\
    & \sum_{k = 1}^{K} \sum_{t=1}^{T_g} \sum_{i=1}^{m_t} \sum_{j=1}^{n_t} \boldsymbol{X}_{ijt}  \mathbbm 1\{\text{item } jt \text{ listed by seller } k\} \geq S_g, \forall g = 1,\cdots,G
    \label{eq:minimum consumer}
    \\
    &\sum_{i=1}^{m_t} \boldsymbol{X}_{ijt} \leq 1, \forall t=1,...,T, \forall jt = 1,...,n_{t} 
    \\
    &\sum_{j=1}^{n_t} \boldsymbol{X}_{ijt} \leq 1, \forall t=1,...,T, \forall it = 1,...,m_{t}  
    \\
    &0 \leq \boldsymbol{X}_{ijt} \leq 1 , \forall t=1,...,T, \forall jt = 1,...,n_{t}, \forall it = 1,...,m_{t}  
    % \text{where } & \tilde{r}_{jtk}  =
    % \begin{cases}
    % r_{jt}, \text{if item } jt \text{ is listed by seller } k,\\
    % 0, else,
    % \end{cases} \\
    % & \tilde{v}_{jtk} =
    % \begin{cases}
    % v_{jt}, \text{if item } jt \text{ is listed by seller } k,\\
    % 0, else. 
    % \end{cases} \\
    % \text{Thus, we have } & r_{jt} = \sum_{k=1}^{K} \tilde{r}_{jtk}, 
    % \text{and } v_{jt} = \sum_{k=1}^{K} \tilde{v}_{jtk}.
\end{align}
\endgroup
where \eqref{eq:lambda const} associates with the relevance constraint, \eqref{eq:inventory} associates with the inventory limit constraint, \eqref{eq:target} associates with the revenue constraint, \eqref{eq:minimum seller} associates with the minimum seller constraint, and \eqref{eq:minimum consumer} associates with the minimum consumer constraint.

Notice that the above LP could be large in terms of the size of the decision variable and the number of constraints, therefore we cannot solve it in real time. Instead, we can collect the data from the last time period and solve the above offline LP. Then, we obtain the optimal dual variables $\xi$ (inventory limit constraint~\eqref{eq:inventory}), $\nu$ (the revenue constraint~\eqref{eq:target}), $\eta$ (the minimum seller constraint~\eqref{eq:minimum seller}), and $\theta$ (minimum consumer constraint~\eqref{eq:minimum consumer}). 

Later, for each incoming impression $t$ induced by consumer $g$, we can modify the previous LP problem proposed in \hyperref[problem: OPT_LP]{$P_1$} by adding extra terms associated with $\xi, \nu, \theta, \eta$ to the objective function. The modified LP problem is illustrated in below:
\begin{align*}
    \max_{\boldsymbol{X}} 
    & 
    \sum_{i=1}^{m_t} \sum_{j=1}^{n_t} h_{it} v_{jt} \boldsymbol{X}_{ijt} 
    - 
    \sum_{k = 1}^{K} \xi_k \sum_{i=1}^{m_t} \sum_{j=1}^{n_t} \boldsymbol{X}_{ijt}  \mathbbm 1\{\text{item } jt \text{ listed by seller } k\} 
    \nonumber
    \\
    &+ 
    \sum_{k=1}^{K} \nu_k \sum_{i=1}^{m_t} \sum_{j=1}^{n_t} h_{it} \boldsymbol{X}_{ijt} v_{jtk} \mathbbm 1\{\text{item } jt \text{ listed by seller } k\} 
    \nonumber
    \\
    &+
    \sum_{k=1}^{K} \eta_k \sum_{i=1}^{m_t} \sum_{j=1}^{n_t} \boldsymbol{X}_{ijt} \mathbbm 1\{\text{item } jt \text{ listed by seller } k\}
    \\
    &+
    \theta_g \sum_{k = 1}^{K} \sum_{i=1}^{m_t} \sum_{j=1}^{n_t} \boldsymbol{X}_{ijt}  \mathbbm 1\{\text{item } jt \text{ listed by seller } k\}
    \\
    \text{s.t. } 
    & \sum_{i=1}^{m_t} \sum_{j=1}^{n_t} h_{it} r_{jt} \boldsymbol{X}_{ijt} \geq \lambda_{t} a_{t}
    \\
    &\sum_{i=1}^{m_t} \boldsymbol{X}_{ijt} \leq 1, \forall j = 1,...,n_{t}  \\
    &\sum_{j=1}^{n_t} \boldsymbol{X}_{ijt} \leq 1, \forall i = 1,...,m_{t}  \\
    &0 \leq \boldsymbol{X}_{ijt} \leq 1, \forall j = 1,...,n_{t}, \forall i = 1,...,m_{t}.
\end{align*}

Notice that the LP formulation can be huge in practice. We present numerical experiments on solving such problems with the recent development of solvers in the next section, which demonstrates the traceability of this approach.

\subsubsection{Computational results for overall planning}

In this section, we present the experiment results for the overall planning problem from Section \ref{sec:overall planning}. The overall planning is formalized as a large-scale problem using the data from a certain period; therefore, it can only be solved offline. Therefore, in this experiment, we formulate the optimization problem using the historical data of an entire day. Specifically speaking, we would set up the restrictions on the existing inventory $L$ and the targeted revenue $I$ for the selected top sellers who have the most listed items. In total, we filter out 5,408 sellers and formulate an optimization problem with 265,646,022 variables, 26,423,572 constraints, and 962,620,739 non-zeros. We use two solvers, which are Gurobi (version 9.0) and OR-Tools \footnote{\url{https://developers.google.com/optimization/lp/lp_advanced}}, and compare their performances in terms of the solving time, number of iterations, and minimum memory spaces required.

Currently, Gurobi adopts three widely-used methods to solve the LP problem, which are respectively primal simplex, dual simplex, and barrier. For OR-Tools, we choose the method of Practical Large-Scale Linear Programming using Primal-Dual Hybrid Gradient (PDLP) \citep{applegate2021practical,applegate2023faster,applegate2024infeasibility}, which specializes in solving large-scale LP problems efficiently. Thus, in this experiment, we incorporate four groups (Primal Simplex, Dual Simplex, Barrier, and PDLP) in total, and compare the performance of PDLP to all other groups. For these experiment groups, under different cores and threads, we evaluate their performance in terms of the solving time. Similar to the previous section, we implement this experiment with the Dell PowerEdge R640 computing node. It is noted that in our case, due to the restricted computing resources, the upper limit of the running time is set to be 48 hours ($172,800$ seconds). 

\vspace{0.2cm}
\begin{table}[ht]
\centering
\begin{threeparttable}
\caption{Experiment Results for Overall Planning} %\vspace{.2in}
\begin{tabular}{ l p{2.1cm} p{2.1cm} p{2.1cm} p{2.1cm}}
    \hline \hline
    \textbf{Solving Method}& PDLP & Barrier & Primal Simplex & Dual Simplex\\
    \hline
   Number of Threads & 4 & 4 & 1* & 1*\\
    Able to Solve the Problem Correctly? & Yes & Yes & Yes & No\\
   Average Solving Time (Sec) & 27,458 & 56,236 & 171,396 & **\\
   Number of Iterations & 3,328 & 208 & 206 & **\\
   Ratio of Cases Reaching the Time Limit & 0 & 0 & 80\%*** & **\\
    \hline \hline
    \end{tabular}
\begin{tablenotes}
\footnotesize
\item Notes. In total, the problem contains 265,646,022 variables, 26,423,572 constraints, and 962,620,739 non-zeros. We run each experiment with the same settings 5 times and take the average of each metric to report. In cases where the solving time limit is reached before the solution is given, we use the time limit(172,800 seconds) as its solving time in terms of calculating the average solving time. * \quad For Primal and Dual simplex methods, Gurobi only allows 1 thread to solve the problem, which cannot be altered. ** \quad Dual simplex fails to give the correct solution, of which the metrics will not be reported. *** \quad Only 1 out of 5 instances finished solving within the time limit. 
\end{tablenotes}
\label{tab:overall planning}
\end{threeparttable}
\end{table}
\vspace{0.2cm}

Out of four groups, we first find that Gurobi Dual simplex method fails to deliver the correct solution due to the numerical sensitivity caused by the excessively large number of variable coefficients, and Gurobi Primal simplex reach the time limit for four out of the five runs. Both PDLP and Gurobi Barrier can solve the problem, while PDLP is averagingly twice as fast as Gurobi Barrier. This demonstrates that solving the huge-scale overall planning problem with the LP-based method is tractable with modern computing architectures and optimization solvers. 

% For the remaining three methods, in table \ref{tab:overall planning}, we show that . In particular, PDLP requires 260GB of memory space, while both Barrier and Primal Simplex methods have a need of 200GB. 
% The average solving time for PDLP is 27,458 seconds, 104\% quicker than Barrier (56,236 seconds) and 504\% quicker than Primal Simplex (171,396 seconds). On the other hand, the results show that the iteration number of PDLP largely exceeds (16 times larger) those of Barrier and Primal Simplex, achieving 3,328, while the other two are 208 and 206 respectively.

% Under the setting where the number of threads is set to 4, all cases in PDLP and Primal Simplex are able to solve the problem correctly in time. Meanwhile, when the number of threads is 1, the primal simplex method fails to solve the problem on time for 80\% of instances. 

\subsection{Representation of the Random Sample from the Full-Scale Experiment }\label{subsec:representation}
We report a corresponding table to Table~\ref{tab:Stat financial performance} in the main draft, constructed using the random sample from the full-scale experiment. As observed, the primary difference lies in the fact that the \textbf{LP90} group does not demonstrate a statistically significant improvement in number of purchases relative to the score-based benchmark.

\vspace{0.2cm}
\begin{table}[ht]
\centering
\begin{threeparttable}
\caption{Representation of the Random Sample} %\vspace{.2in}
\begin{tabular}{ccc}
    \hline \hline
    & LP90 & LP95 \\
    & (Treatment 1) & (Treatment 2) \\
    \hline
    %Revenue per consumer (USD) & 0.0293 &  0.0295  & 0.0290 \\
    \textit{Purchases} & 0.23\% \quad (0.51) & 7.77\%*** \quad (<0.001)  \\
    %\\
    \textit{GMV} & 2.16\% \quad (0.11)  & 6.09\%*** \quad (<0.001) \\
    %\\
    \textit{Revenue} & 3.07\%*** \quad (<0.01) & 6.63\%** \quad (0.01)  \\
    %GMV per consumer (USD) & 0.5553 & 0.5465 & 0.5393\\
    
    %Purchases per consumer & 0.0162 & 0.0159 & 0.0157  \\
    
    %Average transaction price (USD)  & 34.36 & 34.37 & 34.44 \\
    %\textit{Average transaction price (USD)} & -0.23\% & -0.20\% \\
    \hline \hline
    \end{tabular}
\begin{tablenotes}
\footnotesize
\item Notes. We report the relative changes of the two treatment groups (\textbf{LP90} and \textbf{LP95}) compared with control group (\textbf{Benchmark}), i.e., the score-based algorithm for the random sample of the full-scale experiment. Associated $p$-values from pairwise $t$-tests are reported in parentheses. The sample size in Treatment 1 is $N_{\text{treatment1}} = 126,713$, in Treatment 2 is $N_{\text{treatment2}} = 136,759$, and in the Control group is $N_{\text{control}} = 126,875$. *** $p<$ 0.01, ** $p<$ 0.05, * $p<$ 0.1.
\end{tablenotes}
\label{tab:financerandomsample}
\end{threeparttable}
\end{table}
\vspace{0.2cm}

\subsection{Supplementary Evidence on Items Favored by the LP-Based Algorithm}\label{subsec:itemsupplementary}
In this section, we provide supplementary evidence on items favored by the LP-based algorithm using the sampled impression-level data from the field experiment. This dataset contains only aggregated features of sponsored listings. We focus on two key features. The first is $\textit{Log Average Price}$, which captures the logarithm of the average product price across all items in a sponsored listing. The logarithmic transformation mitigates the influence of extreme values. The second feature is $\textit{Weighted PTR}$, a measure of the predicted likelihood of purchase. Formally, it is defined as the inner product of the position weight vector $\boldsymbol{h}$ and the relevance score vector $\boldsymbol{r}$.

We implement the following empirical strategy setup:
\begin{equation}
    Y_{i} = \beta_{0} + \beta_{1} \times \text{Algo90}_i +  \beta_{2} \times \text{Algo95}_i + \eta_t + \epsilon_{i},
\end{equation}
where for a given impression indexed by $i$, $Y_{i}$ is the dependent variable, and $Y_{i} \in \{$\textit{Log Average Price}$, $\textit{Weighted PTR}$\}$. $\text{Algo90}_i, \text{Algo95}_i$ are indicators representing which algorithms the current impression has been assigned to (i.e., \textbf{LP90} or \textbf{LP95}; if both are zero, then it is assigned to \textbf{Benchmark}). $\beta_1$ and $\beta_2$ are the parameters of interest, capturing the effects of \textbf{LP90} and \textbf{LP95} on the dependent variable of interest relative to the score-based algorithm. We include day fixed effects $\eta_t$ to account for any time-related heterogeneity. $\epsilon_{i}$ is the error term. 

\vspace{0.2cm}
\begin{table}[ht]
\centering
\begin{threeparttable}
\caption{Sponsored Listing Characteristics Favored by the LP-based Algorithm.} %\vspace{.2in}
\begin{tabular}{ l p{3.1cm} p{3.1cm}}
    \hline \hline
    & Log Average Price & Weighted PTR \\
    \hline
   LP90 & -0.064*** & 0.117\\
    &(0.002) & (0.091) \\
  %  \\
   LP95 & -0.109*** & 0.483***\\
   & (0.002) & (0.074) \\
    \hline \hline
    \end{tabular}
\begin{tablenotes}
\footnotesize
\item Notes. Observation numbers are $N = 2,071,833$. Robust standard errors are reported in parentheses. *** $p<$ 0.01, ** $p<$ 0.05, * $p<$ 0.1.
\end{tablenotes}
\label{tab:Regression results}
\end{threeparttable}
\end{table}
\vspace{0.2cm}  

Table~\ref{tab:Regression results} reports how the LP-based algorithms (\textbf{LP90} and \textbf{LP95}) alter the characteristics of sponsored listings relative to the score-based benchmark. Both \textbf{LP90} and \textbf{LP95} significantly reduce the $\textit{Log Average Price}$ of displayed items, with \textbf{LP90} lowering it by 0.064 and \textbf{LP95} by an even larger 0.109. This contrasts with the results in Section~\ref{subsubsec:lpfavored}. Specifically, the average product price in our random sample from the full-scale experiment is 2.69 USD, compared to 301.31 USD in the pre-experiment sample, indicating that the two samples differ substantially in product categories. This suggests that the LP-based algorithms prioritize items with higher expected revenue, consistent with our main findings in Section~\ref{subsubsec:lpfavored}, rather than systematically favoring either more expensive or cheaper products. For lower-priced categories, demand effects dominate, leading to the observed price decreases in Table~\ref{tab:Regression results}. Conversely, for higher-priced categories, price effects dominate, resulting in price increases among displayed items in Section~\ref{subsubsec:lpfavored}. Importantly, the observed improvements of the LP-based algorithms cannot be attributed to changes in transaction prices, as shown in Section~\ref{subsec:main_result}. The relative changes in average transaction price are small and statistically insignificant for both treatments (\textbf{LP95}: –0.23\%; \textbf{LP90}: –0.20\%). Regarding $\textit{Weighted PTR}$, \textbf{LP90} shows no statistically significant change, whereas \textbf{LP95} yields a sizable and highly significant increase of 0.483, consistent with the findings in Section~\ref{subsubsec:lpfavored}.

\subsection{The Revenue-Relevance Trade-Off}\label{subsec:revenuerelevance}
In this section, we provide some empirical evidence for the revenue-relevance trade-off. Through our LP model, we can identify the trade-off between revenue and relevance. When the hyperparameter $\lambda$ in problem \hyperref[problem: OPT_LP]{$P_1$} increases, the relevance increases while the optimal revenue decreases. 

From September 13th, 2022 to September, 27th, 2022, we conducted a second A/B test on the same marketplace (dWeb), in which we mainly focused on exploring the comparison of performances for various values of $\lambda$. In total, we have set up four experiment groups, which are named respectively \textbf{LP90}, \textbf{LP85}, \textbf{LP80}, and \textbf{LP75}, taking the associated $\lambda$ to be 0.90, 0.85, 0.80, and 0.75. In total, we count 17,931,775 unique consumers and 199,681,466 impressions. Specifically, \textbf{LP90} was assigned 4,483,091 consumers, generating 49,916,423 impressions; \textbf{LP85} was assigned 4,487,531 consumers, generating 49,983,811 impressions; \textbf{LP80} was assigned 4,479,883 consumers, generating 49,883,948 impressions; \textbf{LP75} was assigned 4,481,270 consumers, generating 49,897,284 impressions. We concentrate on two metrics: \textit{Revenue per consumer} and  \textit{Purchases per consumer} since these are the most essential evaluation metrics for revenue and relevance, as we explained before. 

We set \textbf{LP75} to be the benchmark, of which we compute the relative change of metrics for other groups over.  Figure \ref{fig: Revenue and relevance} shows the apparent trade-off over these two metrics. For \textit{Revenue per consumer}, compared to \textbf{LP75}, the relative changes of \textbf{LP90}, \textbf{LP85}, \textbf{LP80} are 1.24\%, 0.60\%, and 0.30\%. For the \textit{Purchases per consumer}, compared to \textbf{LP75}, the relative changes of \textbf{LP90}, \textbf{LP85}, \textbf{LP80} are -1.19\%, -0.89\%, and 0.11\%. As $\lambda$, the hyperparameter controlling the overall relevance, decreases, \textit{Purchases per consumer} drops as well while \textit{Revenue per consumer} increases. 

\begin{figure}
\centering
  \begin{subfigure}[b]{0.49\textwidth}
    \includegraphics[width=\textwidth]{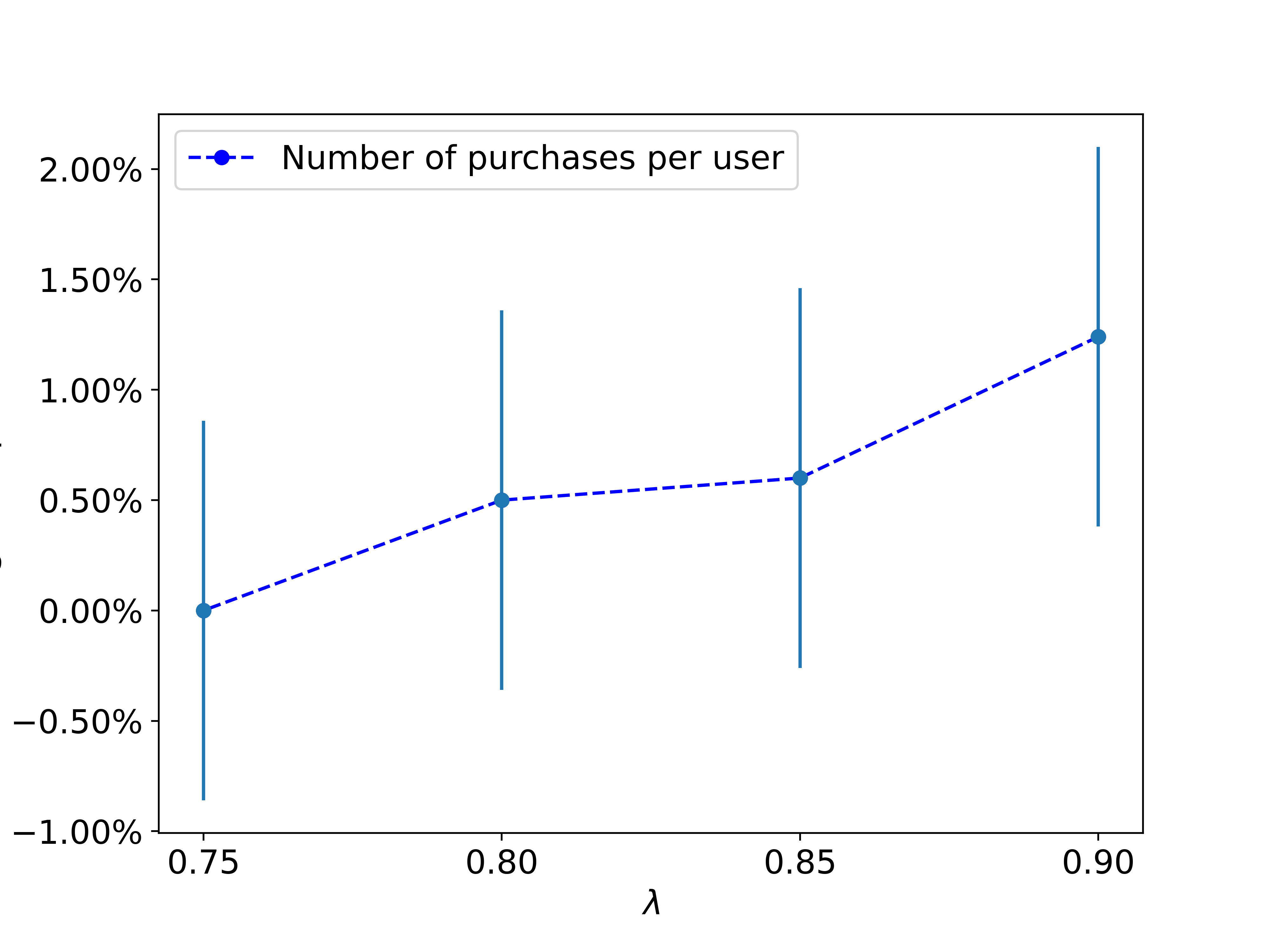}
    \caption{Relative change of number of purchases}
  \end{subfigure}
  \hfill
  \begin{subfigure}[b]{0.49\textwidth}
    \includegraphics[width=\textwidth]{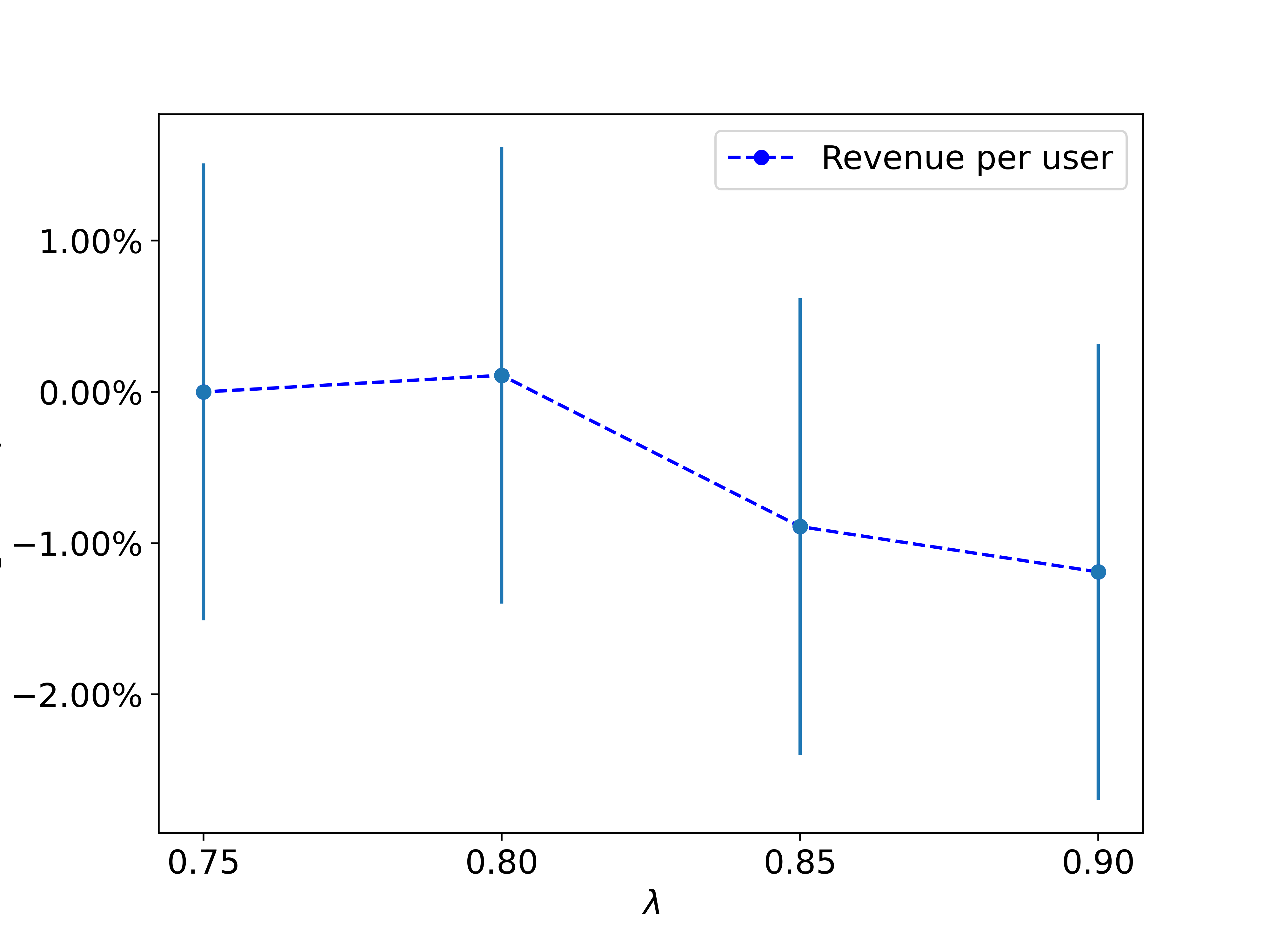}
    \caption{Relative change of revenue}
  \end{subfigure}
  \captionsetup{labelfont=bf}
  \caption{Plots with the error bar showing the comparison of the relative change of (a) the Purchases per consumer and (b) Revenue per consumer for \textbf{LP90}, \textbf{LP85}, \textbf{LP80}, and \textbf{LP75}. \label{fig: Revenue and relevance}}
\end{figure}

\begin{comment}
    \vspace{0.2cm}
\begin{center}
\captionsetup{labelfont=bf}
\captionof{table}{Relative changes for the financial performance over \textbf{LP75}} \label{tab:Stat financial performance exp 2} 
\scalebox{1.0}{
    \begin{tabular}{ccccc}
    \toprule
    & LP90 & LP85 & LP80 \\
    \hline
    %Revenue per consumer (USD) & 0.0785 & 0.0787 &  0.0795  & 0.0805 \\
    \textit{Revenue per consumer (USD)} (p-value) & -1.19\% \quad (0.12) & -0.89\% \quad (0.25) & 0.11\% \quad (0.89)  \\
    %Purchases per consumer & 0.0445 & 0.0442  & 0.0442 & 0.0439  \\
    \textit{Purchases per consumer} (p-value) & 1.24\% \quad (<0.01) & 0.60\% \quad (0.167) & 0.50\% \quad (0.25) \\
    \bottomrule 
    \end{tabular}
}
\begin{tablenotes}
      \small
      \item \textit{Notes.} Most of the demonstrated p-values above do not display evidence for significant differences. The reason for that is largely attributed to the lack of data samples. The closer the mean of the two groups, the larger the amount of samples we would require to show the statistically significant differences. In practice, an A/B test would only run for at most two weeks, not being able to generate enough amount of samples. Nevertheless, instead of focusing on hypothesis testing, the purpose of this experiment is to explore the trade-off of the performances for different groups by changing the hyperparameter $\lambda$.
\end{tablenotes}
\end{center}
\vspace{0.2cm}
\end{comment}

\subsection{Empirical Evidence on Revenue Independence between SLR and KR}\label{subsec:SLRKR}
In this section, we examine whether the observed increases in SLR revenues may partly cannibalize the purchases and revenues of the \textit{seed item}. To investigate this possibility, we draw on the sampled data collected prior to the launch of the field experiment, which includes revenue information for the \textit{seed item}. For each impression, we define KR revenue as the revenue that Marketplace A generates from sales of the \textit{seed item}. Using three empirical tests, we find no evidence that increases in SLR revenue may reduce KR revenue.

\textbf{Test 1: Pearson correlation.} For each impression $i$, we denote KR revenue as $\textit{RevenueKR}_{i}$ and SLR revenue as $\textit{RevenueSLR}_{i}$. Then we compute the Pearson correlation coefficient between $\textit{RevenueKR}_{i}$ and $\textit{RevenueSLR}_{i}$ across all impressions $i = 1,\cdots,266,003$. The resulting correlation coefficient is $0.001$ with p-value of $0.901$, indicating no significant correlation between $\textit{RevenueKR}_{i}$ and $\textit{RevenueSLR}_{i}$.

\textbf{Test 2: Granger causality test.} We conduct a Granger causality test (\citealt{granger1969investigating}) to examine whether the sequence of $\textit{RevenueSLR}{i}$ helps forecast $\textit{RevenueKR}{i}$ at various lags, which would indicate a potential dependency. Since our sampled impressions include timestamps, we are able to order them chronologically and construct the required lag structures. As reported in \cref{tab:Granger Causality Test}, the results reveal no evidence of a significant causal relationship. Across all tested lag lengths (1, 2, 10, and 20), the p-values remain non-significant, suggesting that SLR revenue does not Granger-cause KR revenue.

\vspace{0.2cm}
\begin{table}[ht]
\centering
\begin{threeparttable}
\caption{Granger Causality Test Results} %\vspace{.2in}
\begin{tabular}{ l p{2.8cm} p{2.8cm} p{2.8cm} p{2.8cm}}
    \hline \hline
    & Lag=1 & Lag=2 & Lag=10 & Lag=20\\
    \hline
   SSR F-test & 0.0230 (0.8795) & 0.0221 (0.9781) & 0.0131 (1.0000) & 0.9239 (0.5560)\\
    SSR Chi2-test & 0.0230 (0.8795) & 0.0442 (0.9781) & 0.1311 (1.0000) & 18.5267 (0.5528)\\
   Likelihood Ratio & 0.0230 (0.8795) & 0.0442 (0.9781) & 0.1311 (1.0000) & 18.5157 (0.5535) \\
   Parameter F-test & 0.0230 (0.8795) & 0.0221 (0.9781) & 0.0131 (1.0000) & 0.9239 (0.5560) \\
    \hline \hline
    \end{tabular}
\begin{tablenotes}
\footnotesize
\item Notes. The table presents Granger causality test statistics with corresponding p-values in parentheses for different lags ($1,2,10,20$).
\end{tablenotes}
\label{tab:Granger Causality Test}
\end{threeparttable}
\end{table}
\vspace{0.2cm}  

\begin{comment}

\vspace{0.2cm}
\begin{table}[H]
\begin{center}
\captionsetup{labelfont=bf}
\captionof{table}{Granger causality test results} \label{tab:Granger Causality Test} 
\scalebox{1.0}{
    \begin{tabular}{ccccc}
    \hline
    & Lag $=1$ & Lag $=2$ & Lag $=10$ & Lag $=20$ \\
    \hline
    SSR F-test & 0.0230 \; (0.8795)  & 0.0221 \; (0.9781)  & 0.0131 \; (1.0000)  & 0.9239 \; (0.5560) \\
    SSR Chi2-test & 0.0230 \; (0.8795) & 0.0442 \; (0.9781) & 0.1311 \; (1.0000)  & 18.5267 \; (0.5528) \\
    Likelihood Ratio & 0.0230\; (0.8795) & 0.0442\; (0.9781) & 0.1311\; (1.0000) & 18.5157\; (0.5535) \\
    Parameter F-test & 0.0230\; (0.8795) & 0.0221\; (0.9781) & 0.0131\; (1.0000) & 0.9239\; (0.5560) \\
    \hline 
    \end{tabular}
}
\begin{tablenotes}
      \small
      \item \textit{Notes.} The table presents Granger causality test statistics with corresponding p-values in parentheses for different lags ($1,2,10,20$).
\end{tablenotes}
\end{center}
\vspace{0.2cm}
\end{table}
\end{comment}

\textbf{Test 3: Pooled OLS.} We further conduct an analysis restricted to sessions with identifiable user IDs. Leveraging our panel of 157,154 impressions from 153,275 unique users, we examine how KR revenue relates to the characteristics of SLR and SLR revenue. To estimate this relationship, we employ a pooled Ordinary Least Squares (OLS) regression, clustering standard errors at the user level.

We set up the following specification:
\begin{align}
    \textit{RevenueKR}_{i} 
    &= 
    \alpha + 
    \beta \textit{RevenueSLR}_{i} +  
    \gamma_{1} \textit{LogSeedPrice}_{i}
    \nonumber
    \\
    & \quad  +
    \gamma_{2} \textit{LogSponsoredListingPrice}_{i}  +
    \gamma_{3} \textit{AdRate}_{i} +
    \delta_{d} + \tau_{h} + \epsilon_{i}.
\end{align}
$\textit{LogSeedPrice}_{i}$ denotes the logarithm of the price of the seed item in impression $i$, $\textit{LogSponsoredListingPrice}_{i}$ refers to the average logarithm of the prices of the SLR items in impression $i$, and $\textit{AdRate}{i}$ captures the average ad rate of the SLR items in impression $i$. In addition, we control for day fixed effects ($\delta_d$) and hour fixed effects ($\tau_{h}$). Standard errors are clustered at the user level. Table~\ref{tab:KRSLR} shows that the relationship is not statistically significant, consistent with the results of the previous two empirical tests. Overall, the evidence suggests that in our setting, improvements in SLR revenue do not materially cannibalize the revenue of the \textit{seed item}.

\vspace{0.2cm}
\begin{table}[ht]
\centering
\begin{threeparttable}
\caption{KR Revenue and SLR Revenue.} %\vspace{.2in}
\begin{tabular}{ l p{3.1cm}}
    \hline \hline
    & KR Revenue \\
    \hline
   SLR Revenue & -0.0003 \\
   & (0.000) \\
   LogSeedPrice & 0.1186 \\
   & (1.650) \\
   LogSponsoredListingPrice & -1.8418 \\
   & (2.196) \\
   Ad Rate & 2.0584** \\
   & (0.889) \\
   Observations & 157,154\\
    \hline \hline
    \end{tabular}
\begin{tablenotes}
\footnotesize
\item Notes. Standard errors clustered at the user level are reported in parentheses. *** $p<$ 0.01, ** $p<$ 0.05, * $p<$ 0.1.
\end{tablenotes}
\label{tab:KRSLR}
\end{threeparttable}
\end{table}
\vspace{0.2cm}

\subsection{Stability of Experimental Results: Implication for Long-Term Effects}\label{subsec:longterm}
A key concern in interpreting short-term experiment results is whether the observed treatment effects persist beyond the experimental window or merely reflect transient responses. Since our goal is to infer long-term impacts from a 17-day online A/B test, it is crucial to evaluate whether the treatment effects of the LP-based ranking algorithms (\textbf{LP90} and \textbf{LP95}) remain stable across the duration of exposure. To this end, we conduct empirical tests of the temporal stability of treatment effects.

We begin by dividing the chronologically ordered impression stream into $B \in {10, 20, 40}$ equal-sized exposure bins, indexed from early to late impressions. Within each bin, we compute average outcomes for each algorithm, including the number of purchases and revenue generated for the marketplace. Formally, let $N$ denote the total number of chronologically ordered impressions in the dataset. For a chosen number of bins $B$, we partition the stream as follows:
\begin{equation} \mathcal{B}_{b} = \{t : \lfloor \frac{t N}{B}\rfloor = b\}, \quad b = 0,\cdots,B-1. \end{equation}
so that each exposure bin contains approximately $N/B$ sequential impressions.

Our impression-level dataset consists of randomized assignments to one of three ranking algorithms: the benchmark score-based algorithm and the two LP-based alternatives (\textbf{LP90} and \textbf{LP95}). From these, we compute bin-level average treatment effects (ATEs) by comparing outcomes in each LP-based group against the benchmark. Specifically, let $\bar{Y}_{g,b}$ denote the mean of an outcome variable (either Revenue per Impression or Number of Purchases) for algorithm $g$ in bin $b$. Then the bin-level ATEs are defined as:
\begin{equation} \text{ATE}_{90,b} = \bar{Y}_{\text{LP}90, b} - \bar{Y}_{\text{Benchmark}, b}, \quad \text{ATE}_{95,b} = \bar{Y}_{\text{LP}95, b} - \bar{Y}_{\text{Benchmark}, b}. \end{equation}

Finally, to test for temporal trends in treatment effects, we regress the bin-level ATEs on bin indices:
\begin{equation} \text{ATE}_{g,b} = \beta_{0,g} + \beta_{1,g} b + \epsilon_{g,b}, \quad b = 0,\cdots,B-1. \end{equation}

This specification allows us to assess whether treatment effects are stable over time ($\beta_{1,g} = 0$) or exhibit systematic upward or downward trends.

\vspace{0.2cm}
\begin{table}[ht]
\centering
\begin{threeparttable}
\caption{Stability of Experimental Results.} %\vspace{.2in}
\begin{tabular}{ lcccc}
    \hline \hline
    & \multicolumn{2}{c}{LP90} & \multicolumn{2}{c}{LP95} \\
   Bin Number & Purchases & Revenue & Purchases & Revenue \\
   \hline
  B = 10 & 0.012 & -0.0002 & 0.008 & 0.000 \\
  & (0.007) & (0.001) & (0.011) & (0.000) \\
  B = 20 & 0.009* & 0.000 & 0.006 & 0.000\\
  & (0.006) & (0.000) & (0.006) & ( 0.000) \\
  B = 40 & 0.008 & 0.0004 & -0.008 & -0.0002 \\
  & (0.006) & (0.001) & (0.008) & (0.000) \\
    \hline \hline
    \end{tabular}
\begin{tablenotes}
\footnotesize
\item Notes. Robust standard errors are reported in parentheses. *** $p<$ 0.01, ** $p<$ 0.05, * $p<$ 0.1.
\end{tablenotes}
\label{tab:longterm}
\end{threeparttable}
\end{table}
\vspace{0.2cm}  

Table~\ref{tab:longterm} shows that, across all bin numbers and outcome variables, there is no statistically significant evidence of any systematic trend in the experimental data, suggesting that our experimental results are likely to hold in the long run.

\end{document}